\documentclass[11pt,a4paper]{article}

\usepackage[T1]{fontenc}
\usepackage[utf8]{inputenc}

\usepackage{mathpazo}          
\usepackage[scaled=0.92]{helvet}
\usepackage{microtype}         
\usepackage[a4paper,left=3cm,right=3cm,top=2.7cm,bottom=2.9cm]{geometry}

\usepackage{amsmath,amssymb,amsthm}
\usepackage{mathtools}
\usepackage{enumerate}
\usepackage{tikz}

\usepackage{graphicx}
\usepackage[font=small,labelfont=bf,labelsep=period]{caption}
\usepackage{subcaption}
\usepackage{placeins}          

\usepackage{xcolor}
\usepackage{titlesec}
\usepackage{titling}
\usepackage{authblk}
\usepackage{abstract}
\usepackage{fancyhdr}
\usepackage[numbers,sort&compress]{natbib}

\definecolor{linkcol}{RGB}{20,60,120}
\definecolor{citecol}{RGB}{20,100,70}
\usepackage[
  colorlinks=true,
  linkcolor=linkcol,
  citecolor=citecol,
  urlcolor=linkcol,
  breaklinks=true,
  pdftitle={Incremental Stability and Convergence Properties of Discrete-Time Projected Control Systems},
  pdfauthor={Riccardo Bertollo, S.J.A.M. van den Eijnden, W.P.M.H. Heemels}
]{hyperref}

\titleformat{\section}[hang]
  {\normalfont\large\bfseries}{\thesection}{0.75em}{}
\titleformat{\subsection}[hang]
  {\normalfont\normalsize\bfseries}{\thesubsection}{0.65em}{}
\titleformat{\paragraph}[runin]
  {\normalfont\bfseries}{}{0pt}{}[.]
\titlespacing*{\section}{0pt}{1.7\baselineskip}{0.7\baselineskip}
\titlespacing*{\subsection}{0pt}{1.2\baselineskip}{0.5\baselineskip}
\titlespacing*{\paragraph}{0pt}{0.8\baselineskip}{0.6em}

\fancypagestyle{plain}{\fancyhf{}%
  \fancyfoot[C]{\footnotesize\thepage}}

\pretitle{\begin{center}\LARGE\bfseries}
\posttitle{\par\end{center}\vskip 0.3em}
\preauthor{\begin{center}\large}
\postauthor{\par\end{center}}
\predate{}\postdate{}

\newenvironment{keywords}
  {\vskip 0.4em\begin{quotation}\noindent\small\textbf{Keywords:}\ \ignorespaces}
  {\end{quotation}\vskip 0.4em}

\theoremstyle{plain}
\newtheorem{thm}{Theorem}[section]
\newtheorem{prop}[thm]{Proposition}
\newtheorem{lemma}[thm]{Lemma}
\newtheorem{corollary}[thm]{Corollary}

\theoremstyle{definition}
\newtheorem{definition}[thm]{Definition}
\newtheorem{assumption}[thm]{Assumption}

\theoremstyle{remark}
\newtheorem{remark}[thm]{Remark}

\newcounter{example}[section]
\renewcommand{\theexample}{\thesection.\arabic{example}}
\newenvironment{example}[1]
  {\par\medskip\refstepcounter{example}%
   \noindent\textbf{Example \theexample\ (#1).}\enspace\ignorespaces}
  {\hfill $\circ$\par\medskip}

\newcommand{\proofstep}[1]{\par\smallskip\noindent\textit{#1}\par\nobreak\smallskip}

\newenvironment{proofof}[1]{\par\medskip\noindent {\it Proof of #1:}}{\hfill $\square$\par\medskip}

\newcommand{\K}{\ensuremath{\mathcal{K}}}
\renewcommand{\L}{\ensuremath{\mathcal{L}}}
\renewcommand{\P}{\ensuremath{\mathcal{P}}}
\renewcommand{\S}{\ensuremath{\mathcal{S}}}
\newcommand{\real}{\ensuremath{\mathbb{R}}}

\newcommand{\nat}{\ensuremath{\mathbb{N}}}

\DeclarePairedDelimiter\floor{\lfloor}{\rfloor}

\allowdisplaybreaks[1]

\title{Incremental Stability and Convergence Properties of\\
       Discrete-Time Projected Control Systems%
       \thanks{Funded by the European Union (ERC Advanced Grant, Proacthis, no.\ 10105538).}}

\author[1,2]{Riccardo Bertollo}
\author[1]{S.J.A.M. van den Eijnden}
\author[1]{W.P.M.H. (Maurice) Heemels}
\affil[1]{Control Systems Technology Section, Department of Mechanical Engineering, Eindhoven University of Technology, The Netherlands}
\affil[2]{ICTEAM, Université catholique de Louvain, Belgium}
\date{}

\begin{document}
    \maketitle
    \begin{abstract}
    Projection-based controllers can overcome fundamental limitations of classical linear time-invariant control by modifying the controller's input-output behavior via projection. A key example is given by the hybrid integrator-gain system, a projected integrator, which has recently found successful application in several industrial systems. While prior work on analysis and design of projection-based control systems has primarily focused on the continuous-time setting and non-incremental analysis, a more refined incremental analysis in discrete-time is needed to better reflect actual digital implementation and obtain more accurate (robust) performance assessment. To address this need, this paper considers incremental stability and convergence analysis of discrete-time projection-based control systems. Our first methodology is based on showing that such controllers preserve the quadratic incremental stability of their nominal (unprojected) dynamics, if the projection metric is well-designed. Building on this, we derive a small-gain condition guaranteeing incremental input-to-state stability for interconnections of projected controllers with general nonlinear plants. A second approach is grounded in a direct Lyapunov-based method for verifying incremental stability in input-affine piecewise-smooth systems, which can be seen as an extension of the classical discrete-time Demidovi\v{c} conditions. We illustrate our results through several examples, and demonstrate performance quantification via nonlinear Bode plots, with a special focus on first-order projection elements.
    \end{abstract}


    \hrule
    \vskip 1.5\baselineskip

    \section{Introduction}
    Projection-based controllers hold the potential to overcome fundamental performance limitations of classical linear time-invariant (LTI) control \cite{Deenen2021Projection-Based, vdEijnden26HIGS}. The intuition underlying projection-based controllers is to keep the input-output pair of the controller within carefully designed sets by means of projection. A prominent example is given by the hybrid integrator-gain system (HIGS), which aims at addressing the negative effects of phase lag in LTI integrators by keeping the sign of the integrator's output equivalent to that of its input by way of projecting the output onto a sector set \cite{Deenen2021Projection-Based, Shi2022ANegative}. The potential of projection-based control in general, and HIGS in specific, has been demonstrated in numerous industrial applications, including wafer scanners \cite{vdEijnden20, Heertjes2019Hybrid}, micro-electromechanical systems \cite{Shi2022ANegative}, and medical devices \cite{Beij24}. Other notable approaches that modify properties of LTI controllers by adding a nonlinear or hybrid mechanism, similar in spirit to projection, include sliding-mode control \cite{Perozzi2022Upgrading}, switching control \cite{Lau03}, and reset control \cite{Zaccarian2005First, Aangenent}, to name but a few.
    
    Over the past years, many constructive and practical tools for stability and performance analysis of projection-based control systems have been developed, see, e.g., \cite{Heemels2020Oblique, Shi2022ANegative, Chu23, vdEijndenAut24} and the references therein. Most of these works focus on continuous-time representations. However, any practical control implementation typically involves digital hardware interacting with physical systems, which makes a sampled-data description more desirable. Several recent works consider the discrete-time and sampled-data setting \cite{Sharif2022ADiscrete-Time, Shi25_digital}, but the results are oriented towards non-incremental stability and performance properties such as input-to-state stability (ISS) and $\mathcal{L}_2$-gains. Despite being important, these latter notions provide a worst-case view on robust stability and performance of projection-based control systems, and may be conservative in that respect. 

   To better quantify robustness and performance, the notions of \emph{incremental} (input-to-state) stability \cite{Angeli2002ALyapunov} and \emph{convergence} \cite{Pavlov2005Uniform} have proven useful. Incrementally stable systems are inherently robust in the sense that small perturbations in initial conditions or inputs only lead to small changes in states and outputs. Moreover, incrementally stable and convergent systems asymptotically forget the effect of initial conditions over time, and admit a steady-state response that solely depends on the input. This opens up possibilities for accurately characterizing system performance through LTI-inspired tools such as nonlinear frequency-response functions \cite{PavlovFRF}. While incremental stability and convergence have been studied in various discrete-time settings \cite{Tran2016Incremental, Tran2018Convergence, Jungers2024DiscreteTime, Sharif2024Analysis}, application to general discrete-time projection-based control systems remains limited.

    In this paper, we are, therefore, concerned with incremental stability and convergence analysis of discrete-time projection-based control systems. Building on our preliminary conference results in \cite{BertolloADHS24}, we make three main contributions. First, we show that discrete-time projection-based controllers preserve \emph{quadratic incremental stability (QIS) properties} of their underlying nominal (unprojected) dynamics, if the projection metric is designed appropriately. This result can be seen as a discrete-time equivalent of \cite[Theorem 2]{Heemels2020Oblique}, and plays a key role in our second contribution: the development of a small-gain condition that guarantees incremental input-to-state stability ($\delta$ISS) for interconnections of a projection element and a general nonlinear system. The QIS-preservation and small-gain result combined provide our first methodology to guarantee $\delta$ISS of projection-based closed-loop control systems.
    In our third contribution, we specialize to the class of input affine (nonlinear) systems, allowing us to develop a second methodology to establish $\delta$ISS of a piecewise-smooth closed-loop system (resulting, for instance, from a projection-based controller interconnected with a smooth plant), using an incremental Lyapunov function. The result is based on exploiting the continuity of the projection operator under appropriate assumptions with guarantees for the existence of a $\delta$ISS Lyapunov function, and can be seen as an extension of the discrete-time Demidovi\v{c} conditions \cite[Theorem 14]{Tran2018Convergence} to the case of input-affine piecewise-smooth systems with continuous dynamics. For interconnections with linear systems, these conditions translate into numerically tractable linear matrix inequalities (LMIs). In both our methodologies, we show that the $\delta$ISS property implies the stronger property of convergence, thereby guaranteeing several desirable characteristics such as the fact that periodic inputs lead to periodic outputs having the same fundamental period as the inputs \cite{Pavlov2005Uniform}. We demonstrate applicability of the two different methodologies on both a nonlinear and linear example, and highlight the relevance for performance analysis by constructing ``nonlinear'' Bode-plots.

    This paper extends the result of the preliminary conference version \cite{BertolloADHS24} in several essential ways. We extend the QIS-preservation and small-gain approach to general projection-based controllers, beyond First-Order Projection Elements (FOPEs), requiring new insights for the preservation of QIS. We further extend the results with a new Lyapunov-based methodology, and provide full proofs which were not available previously in \cite{BertolloADHS24}. Throughout various parts of the paper, we relate our results to the special case of FOPEs, which was the object of study in our preliminary work \cite{BertolloADHS24}, and which are of particular interest due to their close connection with HIGS. 

    The remainder of this paper is organized as follows.
    In Section~\ref{sec:incremental-convergence}, we summarize the main definitions of incremental stability and convergence properties.
    In Section~\ref{sec:dynamics}, we define the dynamics of discrete-time projected controllers, and precisely describe the problem setting under consideration: certifying global uniform convergence of discrete-time nonlinear systems including projected controllers.
    In Section~\ref{sec:inheritance-SG}, we identify conditions under which a projected controller inherits incremental properties of its nominal (unprojected) dynamics, and use these conditions to certify the convergence of the closed loop through an incremental small-gain theorem.
    In Section~\ref{sec:PW-dynamics}, we consider an explicit representation of the closed-loop system through switched, piecewise-smooth and continuous dynamics, and certify incremental stability through a quadratic incremental Lyapunov function.
    We conclude with two numerical examples in Section~\ref{sec:simulation}.

    \noindent
    {\bf Notation:} We denote by $\real_{>0} := (0, \infty)$ the set of positive real numbers.
    Given $N \in \nat$, we denote by $\nat_{\geq N}$ the set of natural numbers greater than or equal to $N$.
    $\mathbb B$ is the closed unit ball of appropriate dimension (clear from the context).
    We denote by $\floor{x} := \max \{j \in \nat : j \leq x\}$ the floor operator.
    
    Given a positive definite matrix $P \in \real^{n \times n}$, we denote by $\lambda_m(P), \lambda_M(P)$ its minimum and maximum eigenvalue, respectively; we also denote by $\kappa(P) := \lambda_M(P)/\lambda_m(P)$ its condition number. 
    Given two vectors $x,y \in \real^n$, we denote by $\langle x,y \rangle_P := x^\top P y$.
    We denote by $\|x\|_P := \sqrt{\langle x, x \rangle_P}$ the $P$-norm of $x$.

    Given a set $X \subset \real^n$, we denote by $\partial X$ its boundary and by $\mathrm{int}(X)$ its interior.
    Given a point $y \in \real^n$, we denote by $d_P(y,X) := \inf_{x \in X} \|x-y\|_P$ the weighted distance of $y$ to $X$; if $P$ is not specified, $d(y,X)$ denotes the Euclidean distance of $y$ to $X$.
    Given two sets $X,Y \subset \real^n$, we denote by $d_H(X,Y) := \max \left\{ \sup_{x \in X} d(x,Y), \sup_{y \in Y} d(y,X) \right\}$ the Hausdorff distance between $X$ and $Y$.
    We call ``$[k_1, k_2]$-sector'' the subset of $\real^2$ defined as $\left\{ (x, y) \in \real^2 : k_1 x^2 \leq x y \leq k_2 x^2 \right\}$. We denote by $\overline{xy} \subset \real^n$, called a segment, the convex hull of two points $x,y \in \real^n$, i.e., $\overline{xy} := \{z \in \real^n : z = \lambda x + (1-\lambda) y, \lambda \in [0,1]\}$.
    
    Given a signal $v : \nat \rightarrow \real^n$, we denote $\|v\|_{[j_1, j_2]} := \sup_{j \in \{j_1, \ldots, j_2\}} |v(j)|$, and by $\|v\|_\infty := \sup_{j \in \nat} |v(j)|$ its infinity norm.
    We say that a function $\alpha: \real_{\geq 0} \to \real_{\geq 0}$ is of class $\K$, denoted $\alpha \in \K$, if it is continuous, zero at zero, and strictly increasing.
    We say that $\alpha \in \K_\infty$ if, additionally, $\alpha(s) \to \infty$ when $s \to \infty$.
    We say that a function $\beta$ is of class $\K\L$, denoted by $\beta \in \K\L$, if for any fixed $s$ the function $r \mapsto \beta(r,s)$ is of class $\K$, and for any fixed $r$ the function $s \mapsto \beta(r,s)$ is non-increasing and tends to zero at infinity.
    We denote by $\mathrm{Id}(s) = s$ the identity function.
    We denote the image of a map $F: \real^n \rightrightarrows \real^m$ by $\mathrm{Im}(F) := \{y \in \real^m: y \in F(x) \mbox{ for some } x \in \real^n \}$.

    To lighten the notation, whenever it is clear from the context, we drop the dependence of a discrete-time signal $z$ on the time $j \in \nat$, denoting the current value $z(j)$ and its previous value $z(j-1)$ as $z$ and $z^-$, respectively.

    \section{Preliminaries: Incremental stability and convergence}
    \label{sec:incremental-convergence}

    Before proceeding with the detailed problem statement and technical derivations, we recall here the adopted definitions of incremental stability and convergence properties, as well as the connections between these two related, but not equivalent, concepts.
    We give these definitions for a general system of the form
    \begin{align}
        \label{eq:gen_sys}
        \begin{split}
            x &= F(x^-,w^-), \\
            y &= H(x),
        \end{split}
    \end{align}
    where $x \in \real^{n_x}$ is the state variable and $w \in \real^{n_w}, y \in \real^{n_y}$ are the external input and output, respectively.
    We denote by $x(j,\xi,w)$ the solution to \eqref{eq:gen_sys} at time $j$, evolving from initial condition $x(0) = \xi$ under the input $w: \nat \to \real^{n_w}$.
    Similarly, we denote by $y(j,\xi,w)$ the output trajectory corresponding to $x(j,\xi,w)$.

    \begin{definition}
        \label{def:GAS/GES}
        The system \eqref{eq:gen_sys} is {\it incrementally uniformly globally asymptotically stable} ($\delta$UGAS), if there exists a $\beta \in \K\L$ such that, for all bounded inputs $w: \nat \rightarrow \real^{n_w}$ and all pairs of initial conditions $\xi_1, \xi_2 \in \real^{n_x}$, it holds that
        \begin{align}
            \label{eq:delta_GES}
            | x(j,\xi_1,w) - x(j,\xi_2,w) | \leq \beta(|\xi_1-\xi_2|, j),
        \end{align}
        for all $j \in \nat$.
    \end{definition}

    The definition of $\delta$UGAS can be extended to other well-known stability concepts, such as incremental input-to-state stability (both definitions can be found in \cite{Tran2016Incremental}).

    \begin{definition}
        \label{def:ISS}
        The system \eqref{eq:gen_sys} is {\it incrementally input-to-state stable} ($\delta$ISS), if there exist $\beta \in \K\L$ and $\gamma \in \K$ such that, for all pairs of initial conditions $\xi_1,\xi_2 \in \real^{n_x}$ and all bounded input signals $w_1, w_2: \nat \rightarrow \real^{n_w}$, it holds that
        \begin{align}
            \label{eq:dISS}
            | x(j,\xi_1,w_1) - x(j,\xi_2,w_2) | \leq \beta \left(|\xi_1-\xi_2|, j \right) + \gamma \left(\| w_1 - w_2 \|_\infty\right),
        \end{align}
        for all $j \in \nat$.
        If $\gamma(s) = \bar \gamma s$, for some $\bar \gamma \in \real_{>0}$, the system \eqref{eq:gen_sys} is $\delta$ISS with {\it linear gain} $\bar \gamma$.
    \end{definition}

    Incremental properties, such as $\delta$UGAS and $\delta$ISS, can also be certified in terms of incremental (ISS) Lyapunov functions.
    \begin{definition}
        A function $V: \real^{n_x} \times \real^{n_x} \rightarrow \real$ is called an incremental Lyapunov function for system \eqref{eq:gen_sys}, if there exist $\rho_1, \rho_2, \alpha \in \K_\infty$ such that, for all $x_1, x_2 \in \real^{n_x}$ and all $w \in \real^{n_w}$, it holds that
        \begin{align}
            \label{eq:dISS-Lyapunov}
            \begin{split}
                &\rho_1(|x_1 - x_2|) \leq V(x_1, x_2) \leq \rho_2(|x_1-x_2|), \\
            &V(F(x_1, w), F(x_2, w)) - V(x_1, x_2) \leq -\alpha(V(x_1, x_2)).
            \end{split}
        \end{align}
        Additionally, if there exists $\sigma \in \K_\infty$ such that, for all $x_1, x_2 \in \real^{n_x}$ and all $w_1, w_2 \in \real^{n_w}$, it holds that
        \begin{align}
            \begin{split}
                V(F(x_1, w_1), F(x_2, w_2)) - V(x_1, x_2) \leq -\alpha &(V(x_1, x_2)) \\ &+ \sigma(|w_1 - w_2|),
            \end{split}
        \end{align}
        then $V$ is called a $\delta$ISS Lyapunov function for system \eqref{eq:gen_sys}.
    \end{definition}

    If $F$ is continuous and zero at zero, the existence of an incremental Lyapunov function for system \eqref{eq:gen_sys} implies that the system is $\delta$UGAS \cite[Theorem 5]{Tran2016Incremental}, and the existence of a $\delta$ISS Lyapunov function implies that the system is $\delta$ISS \cite[Theorem 8]{Tran2016Incremental}.
    
    Incremental stability can also be defined in the input-output sense, as follows:
    \begin{definition}
        \label{def:IOS}
        The system \eqref{eq:gen_sys} is {\it incrementally input-to-output stable} ($\delta$IOS), if there exist $\beta \in \K\L$ and $\gamma \in \K$ such that, for all pairs of initial conditions $\xi_1,\xi_2 \in \real^{n_x}$ and all bounded input signals $w_1, w_2: \nat \rightarrow \real^{n_w}$, it holds that
        \begin{align}
            \label{eq:dIOS}
            | y(j,\xi_1,w_1) - y(j,\xi_2,w_2) | \leq \beta \left(|\xi_1-\xi_2|, j \right) + \gamma \left(\| w_1 - w_2 \|_\infty\right),
        \end{align}
        for all $j \in \nat$.
        If $\gamma(s) = \bar \gamma s$, for some $\bar \gamma \in \real_{>0}$, the system \eqref{eq:gen_sys} is $\delta$IOS with {\it linear gain} $\bar \gamma$.
    \end{definition}

    Definitions~\ref{def:ISS} and \ref{def:IOS} are strongly related and, if $H$ is globally Lipschitz, they are equivalent under the additional assumption of \textit{incremental detectability} (see Proposition~\ref{prop:dISS-dIOS} below).
    \begin{definition}
        \label{def:SUO}
        The system \eqref{eq:gen_sys} is {\it incrementally detectable}, if there exist $\beta \in \K\L$ and $\gamma \in \K$ such that, for all pairs of initial conditions $\xi_1,\xi_2 \in \real^{n_x}$ and all bounded input signals $w_1, w_2: \nat \rightarrow \real^{n_w}$, it holds that
        \begin{align}
            \label{eq:dDet}
            |x(j,\xi_1,w_1) - x(j,\xi_2,w_2)| &\leq \beta(|\xi_1-\xi_2|, j) \\
            &\hspace{20pt} + \gamma \left( \|(w_1-w_2,y_1-y_2)\|_\infty \right), \nonumber
        \end{align}
        for all $j \in \nat$.
    \end{definition}
    With the definition of incremental detectability, we can go from $ \delta$IOS to $\delta$ISS, thanks to the following result, which is a discrete-time incremental equivalent of \cite[Proposition 3.1]{Jiang1994Small-gain}. The proof can be found in Appendix~\ref{app:proof-technical}.
    \begin{prop}
        \label{prop:dISS-dIOS}
        Consider system \eqref{eq:gen_sys} with $H$ globally Lipschitz continuous.
        The system \eqref{eq:gen_sys} is $\delta$ISS if and only if it is $\delta$IOS and incrementally detectable.
    \end{prop}
    Lastly, we define the concept of (time-independent) uniform global convergence \cite[Def. 2]{Pavlov2012Steady-State}
    \begin{definition}
        \label{def:convergent}
        The system \eqref{eq:gen_sys} is {\it uniformly globally convergent} if, for every bounded input $w: \mathbb Z \rightarrow \real^{n_w}$,
        \begin{enumerate}
            \item there exists a unique bounded solution $\bar x_w(j)$ to \eqref{eq:gen_sys} that is defined for all $j \in \mathbb Z$
            \item there exists $\beta \in \K\L$ such that, for all $\xi \in \real^{n_x}$, it holds that
            \begin{align}
                \label{eq:convergent}
                |x(j, \xi, w) - \bar x_w(j)| \leq \beta(|\xi - \bar x_w(0)|, j),
            \end{align}
            for all $j \in \nat$.
        \end{enumerate}
    \end{definition}
    From Definition~\ref{def:convergent}, it follows that, for a given bounded input $w$, the steady-state solution to a convergent system is unique.
    Moreover, if $w$ is periodic, the steady-state solution is also periodic with the same period \cite{Pavlov2012Steady-State}.
    This is an important property in the context of control systems.
    Indeed, by showing that a control system is convergent, we enable the use of {\it nonlinear frequency response functions} \cite[Section 3]{Pavlov2012Steady-State}, which can be used to characterize the steady-state response to all harmonic excitations.
    In the homogeneous case, one can also plot these functions against the excitation frequency, obtaining nonlinear Bode-like plots, which can be useful for the design and the performance characterization of the controller (see, e.g., \cite{Heertjes2019Hybrid} for an example in continuous time and Section~\ref{sec:simulation} below for an example in discrete time).

    Incremental stability and convergence are related concepts, although not necessarily equivalent, as extensively investigated in \cite{Tran2018Convergence}.
    However, we have the following result connecting the two concepts, which can be found in \cite[Thm. 13]{Tran2018Convergence}
    \begin{prop}
        \label{prop:incremental-implies-contractive}
        {\bf \cite[Thm. 13]{Tran2018Convergence}}
        If the system \eqref{eq:gen_sys} is $\delta$UGAS, $F$ is continuous and for any bounded input $w$ there exists a compact forward invariant set for \eqref{eq:gen_sys}, then the system \eqref{eq:gen_sys} is uniformly globally convergent.
    \end{prop}
    Using Proposition~\ref{prop:incremental-implies-contractive}, we can show that global uniform convergence is directly implied by $\delta$ISS and some regularity conditions on $F$, as formalized in the next corollary (the proof can be found in Appendix~\ref{app:proof-technical}).
    \begin{corollary}
        \label{cor:dISS-convergence}
        If $F$ in \eqref{eq:gen_sys} is continuous, $F(0,0) = 0$, and \eqref{eq:gen_sys} is $\delta$ISS, then system \eqref{eq:gen_sys} is also uniformly globally convergent.
    \end{corollary}
    

    \section{Projected controllers and problem formulation}
    \label{sec:dynamics}

    Consider a general discrete-time nonlinear plant:
    \begin{align}
        \label{eq:plant-NL}
        &\begin{cases}
            \; x_p = f_p(x_p^-, u^-, w^-), \\
            \; e = h_p(x_p),
        \end{cases}
    \end{align}
    where $x_p \in \real^{n_p}$ is the plant state, $e \in \real^{n_e}$ is the plant tracking error (taking the reference as zero),
    $u \in \real^{n_u}$ is the control input and $w \in \real^m$ is an external disturbance.
    In this paper, we are interested in the analysis and design of a projected controller for the plant \eqref{eq:plant-NL}, which will be specified in Section~\ref{subsec:projection} below.
    In particular, we aim for incremental stability properties of the closed-loop interconnection.

    \subsection{Projected controllers}
    \label{subsec:projection}

    A projected controller consists of a nominal nonlinear controller
    \begin{align}
        \label{eq:controller-NL-nominal}
        \begin{cases}
            \; x_c = f_c(x_c^-, e), \\
            \; u = h_c(x_c),
        \end{cases}
    \end{align}
    with state $x_c \in \real^{n_c}$, and a projection operator, rendering some set $\S_{eu} \subset \real^{n_e} \times \real^{n_u}$ in the $(e,u)$-space forward invariant, i.e., ensuring that $(e(j),u(j)) \in \S_{eu}$ at all times $j \in \nat$, without altering the controller behavior in the interior of $\S_{eu}$.
    The appropriate projection operator can be defined as the solution of a constrained quadratic program (QP), as follows.
    \begin{definition}
        \label{def:projection}
        Given a closed set $\mathcal S \subset \real^{n_c}$, a positive definite matrix $P \in \real^{n_c \times n_c}$ and a point $x \in \real^{n_c}$, the projection {$\Pi^P_\mathcal{S}(x)$} of $x$ on $\mathcal S$ with metric $P$ is defined as
        \begin{align}
            \label{eq:projection}
            \Pi^P_\S (x) := \arg\min\nolimits\limits_{y \in \S} \| y - x \|_P,
        \end{align}
        i.e., the set of closest points to $x$ in $\S$ according to the norm $\|\cdot\|_P$.
    \end{definition}
    For a closed set $\S$, $\Pi^P_\S(x)$ is non-empty due to the Weierstrass theorem, but can contain more than one point.
    If the set $\S$ is additionally convex, $\Pi^P_\S(x)$ is a singleton for all $x$, see, e.g., \cite[Theorem 3.14]{Bauschke2017}, substituting $\|\cdot\|$ with $\|\cdot\|_P$.
    In the remainder of this paper, we consider only projections on closed convex sets, making $\Pi^P_\S(x)$ a singleton, and we consider them not as sets, but as standard functions $\Pi^P_\S: \real^n \to \real^n$.
    
    The new controller dynamics is then given by the following discrete-time \textit{projected dynamical system}
    \begin{align}
        \label{eq:controller-NL-proj}
        \begin{cases}
            \; x_c = \Pi^P_{\S(e)} \left( f_c(x_c^-, e) \right), \\
            \; u = h_c(x_c),
        \end{cases}
    \end{align}
    where the $e$-dependent set $\S(e)$ is given by
    \begin{align}
        \label{eq:Se}
        \S(e) := \left\{ x_c \in \real^{n_c} : (e, h_c(x_c)) \in \S_{eu} \right\},
    \end{align}
    and we make the following assumption on the set $\S(e)$.
    \begin{assumption}
        \label{ass:convex-Se}
        For all $e \in \real^{n_e}$, the set $\S(e)$ is nonempty, closed and convex.
    \end{assumption}
    \begin{remark}
        \label{rmk:error-dependent}
        The particular structure of \eqref{eq:controller-NL-proj}-\eqref{eq:Se}, where the projection operator is defined as a function of $e$, follows from the control application of this framework \cite{Sharif2024Analysis,vdEijnden26HIGS}.
        Indeed, as the value of $e$ represents the error signal coming from the controlled plant \eqref{eq:plant-NL}, it cannot be modified.
        Then, the only way to keep $(u, e)$ inside the set $\S_{eu}$ is to project $u$ on the ($e$-dependent) set $\S(e)$ in \eqref{eq:Se} (see Figure~\ref{fig:proj_difference}).
    \end{remark}
    To give a practical motivation and an illustrative example of the use of projections in control applications, we particularize dynamics \eqref{eq:controller-NL-proj}-\eqref{eq:Se} to the case of the First-Order Projection Element (FOPE) \cite{BertolloADHS24,Sharif2024Analysis}.

    \begin{example}{First Order Projection Elements (FOPEs)}
        Among the various applications of projections in control, the so-called hybrid integrator-gain system (HIGS) \cite{Deenen2021Projection-Based} had quite some engineering successes in motion control of high-precision positioning systems, and it was shown to be able to overcome fundamental limitations of LTI controllers \cite{vDinther2021Overcoming,vdEijnden26HIGS}, thanks to the projection of continuous-time integrator dynamics on the $[0, k_h]$-sector in the input-output plane.
        The rationale behind the choice of such a sector is to ensure that the tracking error and the control output have the same sign at all times, which enhances the performance of the controller.

        In order to generalize this control element and consider discrete-time implementations, the FOPE, presented in \cite{BertolloADHS24}, is defined as a SISO discrete-time controller with scalar state, where linear dynamics are projected on the more general $[k_1,k_2]$-sector.
        Its dynamics are given in terms of \eqref{eq:controller-NL-proj} by
        \begin{align*}
            \begin{cases}
                x_c = \Pi_{\S(e)} \left( f_c(x_c^-, e) \right) = \Pi_{\S(e)} \left( a x_c^- + b e \right), \\
                u = x_c,
            \end{cases}
        \end{align*}
        or, in shorthand,
        \begin{align}
            \label{eq:FOPE}
            u = \Pi_{\S(e)} ( \underbrace{a u^- + b e}_{=: v} ),
        \end{align}
        where $a \in \real$, $b \in \real_{>0}$ are controller parameters, and $v$ describes the right-hand side of the unprojected linear dynamics.
        The output of the FOPE is equal to its state $u$.
        Since, in this example, the desired invariant set $\S_{eu}$ is the $[k_1,k_2]$-sector, where $k_1, k_2 \in \real, \; k_1 < k_2$ are two additional control parameters, the error-dependent set $\S(e)$ in \eqref{eq:Se} is given by
        \begin{align}
            \label{eq:Se-FOPE}
            \begin{split}
                \S(e) := &\left\{ u \in \real : k_1 e^2 \leq ue \leq k_2 e^2 \right\} \\ =
                &\begin{cases}
                    [k_1 e, k_2 e], &\mbox{when } e \geq 0, \\
                    [k_2 e, k_1 e], &\mbox{when } e < 0.
                \end{cases}
            \end{split}
        \end{align}
        Figure~\ref{fig:proj_difference} illustrates an example of the sets $\S_{eu}$ and $\S(e)$, as well as the key difference between the projection of $v$ on the input-dependent set $\S(e)$ and the projection of the pair $(e,u)$ on the set $\S_{eu}$ (see the discussion in Remark~\ref{rmk:error-dependent}).
        Notice, moreover, that while the $[k_1,k_2]$-sector is not convex, each set $\S(e)$, being a segment, is convex, and therefore it satisfies Assumption~\ref{ass:convex-Se}.
        
        Note that the discrete-time HIGS, proposed in \cite{Sharif2024Analysis}, is a particular case of the FOPE in \eqref{eq:FOPE}, where $a = 1$ (integrator dynamics), $k_1 = 0$ and $k_2 = k_h > 0$.
    \end{example}
    
    The example above shows that it is not necessary for $\S_{eu}$ to be convex in order for $\S(e)$ to be convex (as required by Assumption~\ref{ass:convex-Se}). Convexity of $\S_{eu}$ is also, in general, not sufficient for convexity of $\S(e)$ (consider, e.g., the convex set $\S_{eu} := \{(e,u) \in \real^2 : u \leq e\}$ and the function $h_c: \real^2 \to \real$ defined by $h_c(x) := x_1 x_2$).

    \begin{figure}[!htb]
        \centering
        \includegraphics[width=0.62\textwidth]{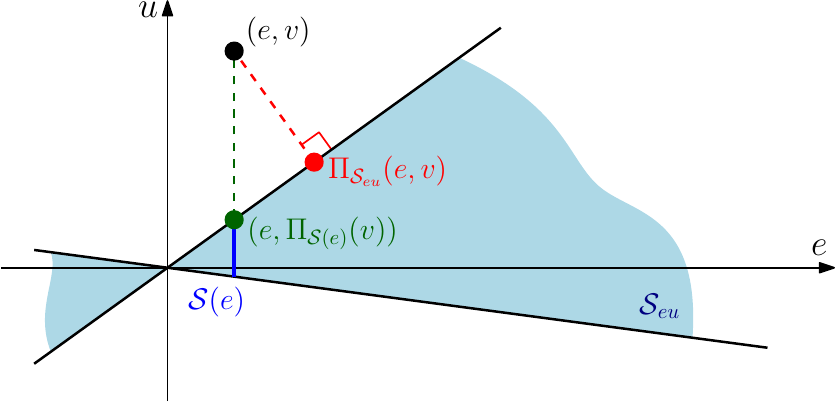}
        \caption{Difference between projecting the pair $(v,e)$ on the set $\S_{eu}$ (red) and projecting $v$ on the input-dependent set $\S(e)$, as per \eqref{eq:FOPE} (green).}
        \label{fig:proj_difference}
    \end{figure}

    \subsection{Problem formulation}
    \label{subsec:problem}

    In this paper, we consider the closed-loop interconnection
    \begin{align}
        \label{eq:CL-NL}
        \begin{cases}
            x = \begin{bmatrix}
                x_p \\ x_c
            \end{bmatrix} = f(x^-, w^-) := \begin{bmatrix}
                f_p(x_p^-, u^-, w^-), \\
                \Pi^P_{\S(e)} \left( f_c(x_c^-, e) \right)
            \end{bmatrix}, \\
            \; e = h_p(x_p), \\
            \; u = h_c(x_c),
        \end{cases}
    \end{align}
    with state $x \in \real^{n_p} \times \real^{n_c}$ and $f_p, h_p, f_c, h_c,$ and $\S: \real^{n_e} \rightrightarrows \real^{n_c}$ coming from \eqref{eq:plant-NL}, \eqref{eq:controller-NL-proj}, \eqref{eq:Se}.

    \begin{remark}
        The reader might notice that, while in \eqref{eq:CL-NL} the plant dynamics is fully explicit, the controller dynamics also depends on the current error $e(j)$, with $j$ the current time.
        This is because, as mentioned previously, the goal of the projection is to keep $(e(j), u(j))$ in the set $\S_{eu}$.
        The complete dynamics \eqref{eq:CL-NL} could be written down explicitly by using the relation $e = h_p(f_p(x_p^-, h_c(x_c^-), w^-))$, as we will illustrate in the example in Section~\ref{sec:PW-dynamics} for the special case where all the maps in \eqref{eq:CL-NL} are linear.
        For the general nonlinear case, we preferred using the form in \eqref{eq:CL-NL}, to avoid overloading the notation.
    \end{remark}

    In addition to the convexity of $\S(e)$, $e \in \real^{n_e}$, as per Assumption~\ref{ass:convex-Se}, we make the following additional assumption on the regularity of the projection of the image of $f_c$ on $\S(e)$:
    \begin{assumption}
        \label{ass:regularity-Se}
        There exists $\alpha_{\S} \in \K_\infty$ such that, for all $x_c \in \mathrm{Im}(\S)$, and all $e_1, e_2 \in \real^{n_e}$,
        \begin{align}
            \label{eq:LH}
            \left| \Pi^P_{\S(e_1)}(f_c(x_c, e_2)) - \Pi^P_{\S(e_2)}(f_c(x_c, e_2)) \right| \leq \alpha_\S(|e_1 - e_2|).
        \end{align}
    \end{assumption}
    While the bound in \eqref{eq:LH} is the one that we will use in the derivations below, the inequality might not be easy to verify, since it requires the solution of the projection QP \eqref{eq:projection}. In the next lemma (whose proof can be found in Appendix~\ref{app:proof-technical}), we identify three conditions, related only to $f_c$ and $\S$, which ensure that Assumption~\ref{ass:regularity-Se} is satisfied.
    \begin{lemma}
        \label{lem:regularity-S}
        Consider the projection-based controller \eqref{eq:controller-NL-proj}-\eqref{eq:Se}.
        Assumption~\ref{ass:regularity-Se} is satisfied if one (or more) of the following conditions holds:
        \begin{enumerate}[{\it (i)}]
            \item \ [$\S$ is a pure translation] for all $e \in \real^{n_e}$, $\S(e) := \S_0 + f_\S(e)$, where $\S_0 \subset \real^{n_c}$ is a non-empty, closed, convex set, and $f_\S: \real^{n_e} \to \real^{n_c}$ is globally Lipschitz;
            \item \ [$\S$ is a convex polyhedron] for all $e \in \real^{n_e}$, $\S(e) := \{x_c \in \real^{n_c} : Ax_c \leq b(e)\}$, where $A \in \real^{n_\S \times n_c}$ and $b: \real^{n_e} \to \real^{n_\S}$ is globally Lipschitz;
            \item \ [global boundedness] there exist $M_c, M_\S > 0$ and $\alpha_H \in \K_\infty$ such that, for all $e_1, e_2 \in \real^{n_e}$ and all $x_c \in \mathrm{Im}(\S)$,
            \begin{subequations}
                \begin{align}
                    &\left\| f_c(x_c, e_2) \right\| \leq M_c, \label{eq:fc-unif-bounded} \\
                    &\S(e_1) \subseteq M_\S \mathbb B, \label{eq:Se-unif-bounded} \\
                    &d_H(\S(e_1), \S(e_2)) \leq \alpha_H(|e_1 - e_2|). \label{eq:dH-Kinf}
                \end{align}
            \end{subequations}
        \end{enumerate}
    \end{lemma}
    Lastly, we assume sufficient regularity of the maps in the right-hand side of \eqref{eq:CL-NL}:
    \begin{assumption}
        \label{ass:regularity}
        The functions $f_p, h_p, f_c, h_c$ in \eqref{eq:CL-NL} are equal to zero when all their arguments are zero.
        Moreover, $f_p, f_c$ are continuous and $h_p,h_c$ are globally Lipschitz.
    \end{assumption}

    \begin{remark}
        It is fair to underline that the three conditions in Lemma~\ref{lem:regularity-S} are not mild.
        In particular, items \textit{(i)-(ii)} require a rigid translation or inflation of the set, while condition \textit{(iii)} requires global boundedness of both $f_c$ and $\S$ (as a set-valued mapping).
        However, these conditions are not conservative: we can show, with a relatively simple example, that an unbounded $f_c$ (item \textit{(iii)} not satisfied), combined with a set $\S(e)$ that rotates as $e$ varies (items \textit{(i)-(ii)} not satisfied), can make it impossible to find $\alpha_\S$ satisfying \eqref{eq:LH}.
        \\
        Consider $e \in \real$, and suppose $\S(e) \subset \real^2$ is a square with side length $2 \ell$, centered in 0 and rotated by $e$ counter-clockwise, i.e.,
        \begin{align*}
            \S(e) := \{ x \in \real^2 : \; &|x_1 \cos (e) + x_2 \sin (e)| \leq \ell \mbox{ and } \\ &|-x_1 \sin (e) + x_2 \cos (e)| \leq \ell \}.
        \end{align*}
        Consider now the controller dynamics
        $$f_c(x_1, x_2) := \begin{bmatrix}
            x_1 \\ \frac{1}{1 - x_2/\ell} - 1
        \end{bmatrix},$$
        and its projection on $\S(e)$ according to the Euclidean norm, i.e., $P = I$.
        Let $p := (0, p_2),$ with $p_2 \in (\ell^2/(\ell+1), \ell)$.
        For $e = 0$, we always have $\Pi_{\S(0)}(f_c(p)) = (0, \ell)$.
        Conversely, for any $e \in (0, \pi/4)$, there always exists $p_2$ such that the projection ends up in the ``top-right corner'' of the rotated square.
        In particular, this happens if $1/(1-p_2/\ell)-1 > \ell/\sin(e)$.
        Combining this fact with \eqref{eq:LH}, we have that the function $\alpha_\S$ has to satisfy $\alpha_\S(e - 0) \geq \ell$ for all $e \in (0, \pi/4)$.
        However, since $\alpha_\S \in \K_\infty$, it needs to be continuous and such that $\alpha_\S(0) = 0$, which is impossible, meaning that Assumption~\ref{ass:regularity-Se} is not satisfied.
    \end{remark}

    Our goal is to identify sufficient conditions to prove that the system \eqref{eq:CL-NL} is $\delta$ISS, motivated by two important consequences of this property.
    First, $\delta$ISS shows that the closed loop is not sensitive to small changes in the initial states or the inputs, see \eqref{eq:dISS}.
    Second, global uniform convergence follows directly from $\delta$ISS of \eqref{eq:CL-NL} and Assumptions~\ref{ass:convex-Se}-\ref{ass:regularity}.
    This is due to Corollary~\ref{cor:dISS-convergence}, and the fact that, under such assumptions, the projection $\Pi^P_{\S(e)}$ inherits the continuity of its argument, which we will prove in the next section. We recall that having the convergence property guarantees the existence of a unique steady-state solution, and that periodic inputs lead to periodic outputs, see Section~\ref{sec:incremental-convergence}. These aspects are key in performance assessment through, e.g., nonlinear Bode-plots, as we will show later in an example.

    We will present in Sections~\ref{sec:inheritance-SG}-\ref{sec:PW-dynamics} two alternative methodologies to certify $\delta$ISS of \eqref{eq:CL-NL}.
    The first methodology is based on the proof that the projected dynamics \eqref{eq:controller-NL-proj}-\eqref{eq:Se} inherits incremental \textit{quadratic} properties of the nominal (unprojected) dynamics \eqref{eq:controller-NL-nominal} if the projection operator is well-designed; this fact is used in combination with {\it incremental small-gain arguments}, to certify $\delta$ISS of the closed-loop interconnection \eqref{eq:CL-NL}.
    We will call this the ``QIS preservation and small-gain approach''.
    The second approach is based on the explicit expression of dynamics \eqref{eq:CL-NL} as a piecewise-smooth system, and on the direct construction of an {\it incremental Lyapunov function} for the closed-loop system exploiting its continuity.
    We will call this the ``closed-loop $\delta$ISS Lyapunov function approach''.

    \section{QIS preservation and small-gain approach}
    \label{sec:inheritance-SG}

    \subsection{Preservation of incremental quadratic ISS}
    \label{subsec:inheritance}
    In this section, we prove that the projected dynamical system \eqref{eq:controller-NL-proj} preserves certain quadratic incremental stability properties of the nominal system \eqref{eq:controller-NL-nominal}.
    We start by showing that, under Assumptions~\ref{ass:convex-Se}-\ref{ass:regularity}, the right-hand side of \eqref{eq:controller-NL-proj}, and therefore the whole closed-loop dynamics $f$ in \eqref{eq:CL-NL}, is continuous. The proof of the lemma can be found in Appendix~\ref{app:proof-technical}.
    \begin{lemma}
        \label{lem:Lipschitz}
        Under Assumptions~\ref{ass:convex-Se}-\ref{ass:regularity}, $f$ in \eqref{eq:CL-NL} is continuous.
    \end{lemma}
    Next, we will show that, if the function
    \begin{align}
        \label{eq:Lyapunov-open-loop}
        V(x_1,x_2) := \|x_1 - x_2\|_P,
    \end{align}
    is a $\delta$ISS Lyapunov function for the nominal controller \eqref{eq:controller-NL-nominal}, then the same is true for the projected controller \eqref{eq:controller-NL-proj}.
    Hence, the $\delta$ISS Lyapunov function is ``inherited'' by the projected version.
    This result can be seen as a discrete-time equivalent of \cite[Theorem 2]{Heemels2020Oblique}.

    \begin{thm}
        \label{thm:dISS-inheritance}
        Suppose that Assumptions~\ref{ass:convex-Se}-\ref{ass:regularity} hold.
        If $V$ in \eqref{eq:Lyapunov-open-loop} is a $\delta$ISS Lyapunov function for the nominal controller \eqref{eq:controller-NL-nominal}, i.e., there exist $\alpha \in (0,1)$ and $\sigma \in \K$ such that
        \begin{align}
            \label{eq:IQISS-Lyapunov}
            \|f_c(x_1, e_1) - f_c(x_2, e_2)\|_P \leq \alpha \|x_1 - x_2\|_P + \sigma(|e_1 - e_2|),
        \end{align}
        then $V$ is also a $\delta$ISS Lyapunov function for the projected controller \eqref{eq:controller-NL-proj}, provided that the same norm $P$ is used in both definitions.
    \end{thm}
    \begin{proof}
        Let us denote $f_{c,i} := f_c(x_i, e_i), \; i \in \{1,2\}$.
        We can use \eqref{eq:LH}, \eqref{eq:IQISS-Lyapunov} and the convexity of $S(e)$, from Assumption~\ref{ass:convex-Se}, to obtain
        \begin{align*}
            \left\| \Pi^P_{\S(e_1)} f_{c,1} \right. & - \left. \Pi^P_{\S(e_2)} f_{c,2} \right\|_P= \\
            &\hspace{-20pt}= \left\| \Pi^P_{\S(e_1)} f_{c,1} - \Pi^P_{\S(e_1)} f_{c,2} + \Pi^P_{\S(e_1)} f_{c,2} - \Pi^P_{\S(e_2)} f_{c,2} \right\|_P \\
            &\hspace{-20pt}\leq \left\| \Pi^P_{\S(e_1)} f_{c,1} - \Pi^P_{\S(e_1)} f_{c.2} \right\|_P + \left\| \Pi^P_{\S(e_1)} f_{c,2} - \Pi^P_{\S(e_2)} f_{c,2} \right\|_P \\
            &\hspace{-20pt}\leq \alpha \left\| x_1 - x_2 \right\|_P + \sigma \left( \left| e_1 - e_2 \right| \right) + \sqrt{\lambda_M(P)} \alpha_{\mathcal S}(|e_1-e_2|) \\
            &\hspace{-20pt}\leq \alpha  \left\| x_1 - x_2 \right\|_P + \gamma \left( \left| e_1 - e_2 \right| \right),
        \end{align*}
        where $\gamma(s) := \left(\sigma + \sqrt{\lambda_M(P)} \alpha_{\mathcal S} \right)(s)$ is of class $\K_\infty$, which proves $\delta$ISS of the projected dynamics \eqref{eq:controller-NL-proj}.
    \end{proof}

    An immediate consequence of the result above, obtained by setting $e_1 = e_2$ is that $\delta$UGAS of the nominal dynamics \eqref{eq:controller-NL-nominal}, certified by $V$ in \eqref{eq:Lyapunov-open-loop}, is inherited by the projected dynamics \eqref{eq:controller-NL-proj}.

    \begin{corollary}
        If $V$ in \eqref{eq:Lyapunov-open-loop} is an incremental Lyapunov function for the nominal controller \eqref{eq:controller-NL-nominal}, i.e., there exists $\alpha \in (0,1)$ such that
        \begin{align}
            \label{eq:IQS-Lyapunov}
            \|f_c(x_1, e) - f_c(x_2, e)\|_P \leq \alpha \|x_1 - x_2\|_P,
        \end{align}
        then $V$ is also an incremental Lyapunov function for the projected controller \eqref{eq:controller-NL-proj}, provided that the same norm $P$ is used in both definitions.
    \end{corollary}

    \begin{example}{$\delta$ISS of FOPEs}
        We can use Theorem~\ref{thm:dISS-inheritance} to show that the FOPE \eqref{eq:FOPE}-\eqref{eq:Se-FOPE} is $\delta$ISS whenever $|a| < 1$.
        Indeed, the set $\S(e)$ in \eqref{eq:Se-FOPE} can be written as a polyhedral set as in item \emph{(ii)} of Lemma~\ref{lem:regularity-S}, by selecting
        \begin{align*}
            A := \begin{bmatrix}
                1 \\ -1
            \end{bmatrix}, \quad\quad b(e) := \begin{cases}
                \big[
                    k_2 e, \; -k_1 e
                \big]^\top, &\mbox{when } e \geq 0, \\
                \big[
                    k_1 e, \; -k_2 e
                \big]^\top, &\mbox{when } e < 0.
            \end{cases}
        \end{align*}
        Consequently, Assumption~\ref{ass:regularity-Se} is satisfied.
        We can confirm this by explicitly computing $\alpha_H$ in \eqref{eq:LH} for this simple case, yielding $\alpha_H(s) := \bar k s := \max\{k_1, k_2\}s.$
        As the FOPE is a scalar-state system, we can use the incremental Lyapunov function $V(x_1, x_2) := |x_1 - x_2|$.
        Using dynamics \eqref{eq:FOPE} and the triangle inequality, we get
        \begin{align}
            \label{eq:FOPE-Lipschitz}
            |x_1 - x_2| \leq |a| |x_1^- - x_2^-| + b |e_1 - e_2|,
        \end{align}
        showing that $V$ is a $\delta$ISS Lyapunov function for the unprojected dynamics and thus, in view of Theorem~\ref{thm:dISS-inheritance}, the FOPE is $\delta$ISS.
        The (linear) incremental gain $\gamma$ in \eqref{eq:dISS} can be computed explicitly in this case.
        Indeed, using \eqref{eq:FOPE-Lipschitz} and the expression of $\alpha_H$ we just computed, we can use similar passages to the proof of Theorem~\ref{thm:dISS-inheritance} to conclude that
        \begin{align*}
            |x_1 - x_2| \leq |a| |x_1^- - x_2^-| + (b + \bar k) |e_1 - e_2|.
        \end{align*}
        By iterating the inequality backwards, as per standard Lyapunov results, we conclude that \eqref{eq:dISS} holds with $\beta(r,j) = |a|^j r$ and linear gain $\gamma = \frac{b + \bar k}{1 - |a|}$.
        Note that this recovers the result \cite[Theorem 3]{BertolloADHS24} as a special case.
    \end{example}

\subsection{Small-gain theorem}
\label{subsec:small-gain}

    In this section, we identify conditions under which the output interconnection of two $\delta$ISS systems is still $\delta$ISS.
    Such results, known as \textit{small-gain theorems}, are well-known in the non-incremental discrete-time setting \cite{Jiang2001Input-to-state}.
    Moreover, recently, two discrete-time incremental small-gain theorems have been proposed in \cite{Schimperna2025OnIncremental}, for the \textit{full-state} interconnection of two $\delta$ISS systems.
    With Theorem~\ref{thm:small-gain} is this section, we extend those results to the case of output feedback between $\delta$ISS subsystems.

    In the following, given a pair of initial conditions $\xi_1,\xi_2 \in \real^{n_p} \times \real^{n_c}$ and two input signals $w_1,w_2: \nat \to \real^{n_w}$, we will use the shorthand notation $\delta x := x(j,\xi_1,w_1) - x(j,\xi_2,w_2)$, and similarly for all the other coordinates.

    \begin{thm}
        \label{thm:small-gain}
        Consider system \eqref{eq:CL-NL}, given by the output interconnection of \eqref{eq:plant-NL} and \eqref{eq:controller-NL-proj}.
        Suppose that \eqref{eq:plant-NL}, \eqref{eq:controller-NL-proj} are incrementally detectable and $\delta$IOS.
        Let functions $\beta_i \in \K\L$ and $\gamma_i, \gamma^w \in \K$, $i \in \{1,2\}$ be such that the $\delta$IOS condition \eqref{eq:dIOS} can be specified for \eqref{eq:plant-NL}, \eqref{eq:controller-NL-proj} in the form
        \begin{align}
            \label{eq:dIOS-particular}
            &\begin{cases}
                |\delta e(j)| \leq \max \{ \beta_1(|\delta \xi_p|,j), \gamma_1(\|\delta u\|_\infty), \gamma^w(\|\delta w\|_\infty) \}, \\
                |\delta u(j)| \leq \max \{ \beta_2(|\delta \xi_c|,j), \gamma_2(\|\delta e\|_\infty) \}.
            \end{cases}
        \end{align}
        If the gain functions $\gamma_1, \gamma_2$ in \eqref{eq:dIOS-particular} satisfy the small-gain condition, i.e., for all $s > 0$,
        \begin{align}
            \label{eq:small-gain}
            \begin{cases}
                \gamma_1 \circ \gamma_2(s) < s, \\
                \gamma_2 \circ \gamma_1(s) < s,
            \end{cases}
        \end{align}
        then the closed-loop system \eqref{eq:CL-NL} is $\delta$ISS.
    \end{thm}

    The proof of Theorem~\ref{thm:small-gain} can be found in Appendix~\ref{app:proof-small-gain}.

    \begin{remark}
        Note that one can easily pass from the formulation of $\delta$IOS based on a summation, as in \eqref{eq:dIOS}, to one based on the maximum, as in \eqref{eq:dIOS-particular}, by using the fact that, for any pair of non-negative numbers, the following inequality holds for any $\rho \in \K_\infty$:
        \begin{align*}
            a + b \leq \max \left\{ \left( \mathrm{Id} + \rho \right) (a), \left(\mathrm{Id}+\rho^{-1} \right)(b) \right\}.
        \end{align*}
        This can be easily seen by considering the two cases $\rho(a) \leq b$ and $\rho(a) > b$.
        In particular, \eqref{eq:dIOS-particular}-\eqref{eq:small-gain} are satisfied if there exist functions $\hat \beta_i \in \K\L$ and $\hat \gamma_i, \hat \gamma^w, \rho_i, \in \K_\infty$, $i \in \{1,2\}$, such that
        \begin{align}
            \label{eq:dIOS-particular-sum}
            &\begin{cases}
                |\delta e(j)| \leq \hat \beta_1(|\delta \xi_p|,j) + \hat \gamma_1(\|\delta u\|_\infty) + \hat \gamma^w(\|\delta w\|_\infty), \\
                |\delta u(j)| \leq \hat \beta_2(|\delta \xi_c|,j) + \hat \gamma_2(\|\delta e\|_\infty),
            \end{cases} \\
            \label{eq:small-gain-sum}
            &\begin{cases}
                (\mathrm{Id} + \rho_1) \circ \hat \gamma_1 \circ (\mathrm{Id} + \rho_2) \circ \hat \gamma_2(s) < s, \\
                (\mathrm{Id} + \rho_2) \circ \hat \gamma_2 \circ (\mathrm{Id} + \rho_1) \circ \hat \gamma_1(s) < s,
            \end{cases} \mbox{for all } s > 0.
        \end{align}
    \end{remark}

\section{Closed-loop $\delta$ISS Lyapunov function approach}
\label{sec:PW-dynamics}

Since a projection is the result of a constrained (parametric) quadratic minimization problem (see Definition~\ref{def:projection}), it is often possible to express $\Pi^P_\S(f_c(x,w))$ as a complementarity system (using the Karush-Kuhn-Tucker sufficient conditions for optimality), or as an explicit (piecewise-smooth) function of the state $x$ and the input $w$.
Examples based on complementarity models can be found in, e.g., \cite{Heemels2000Projected}, but we are not aware of any result in the literature that analyzes the incremental stability properties of discrete-time complementarity systems (notably, the recent work \cite{Raghunathan2025Stability} focuses on stability of such systems, but not incremental stability).
Examples of the explicit formulation as piecewise-smooth dynamics can be found, e.g., in \cite{Sharif2024Analysis} and in the example at the end of this section, but also appear in explicit model predictive control using parametric (quadratic) programming \cite{Bemporad2002OnHybrid}.
For the incremental stability analysis of this type of dynamics, we can take inspiration from works such as \cite{Pavlov2008Convergent}. Therefore, we choose the piecewise-smooth formalism for the analysis in this section.

In particular, we consider the case where the closed-loop dynamics \eqref{eq:CL-NL} can be written as a piecewise-smooth system with input-affine dynamics, namely
\begin{align}
    \label{eq:input-affine}
    \begin{split}
        x = F_i(x^-, w^-) := f_i(x^-) + g_i(x^-) w^-, \quad (x^-, w^-) \in \Omega_i, \\ \; i \in \{1, \ldots, q\},
    \end{split}
\end{align}
where $x \in \real^n$, with $n = n_p + n_c$, $w \in \real^m$.
We assume that the sets $\Omega_i, \; i \in \{1, \ldots, q\},$ are closed, and that they form a partition of $\real^{n+m}$, i.e., $\bigcup_{i \in \{1, \ldots, q\}} \Omega_i = \real^{n+m}$ and $\mathrm{int}\left(\Omega_i\right) \cap \mathrm{int}\left(\Omega_j\right) = \emptyset$ for all $i,j \in \{1, \ldots, q\}, i \neq j$.
Moreover, the functions $f_i, g_i$ are defined on an open set containing $\Omega_i$, they are $\mathcal{C}^1$ on their domain, and if $(x,w) \in \Omega_i \cap \Omega_j$ then $F_i(x,w) = F_j(x,w)$.
Hence, the right-hand side of \eqref{eq:input-affine} is continuous and piecewise-smooth.
Moreover, we make the following assumption on the regularity of the sets $\Omega_i$.
\begin{assumption}
    \label{ass:Omega_i}
    For all $p^a, p^b \in \real^{n+m}$, there is a finite number $\ell$ of real numbers $0 = \lambda_1 < \lambda_2 < \ldots < \lambda_\ell = 1$, defining the same number of points $p_k = p^a + \lambda_k \left(p^b - p^a\right), \; k \in \{1,2,\ldots,\ell\}$, such that, for all $k \in \{1,2,\ldots,\ell-1\}$, the segment $\overline{p_k p_{k+1}} \subset \Omega_{h(k)}$, for some $h(k) \in \{1,2,\ldots,q\}$.
\end{assumption}
The purpose of Assumption~\ref{ass:Omega_i} is to exclude pathological cases, where a segment between two points would intersect the boundaries between the sets $\Omega_i$ infinitely many times. One example where Assumption~\ref{ass:Omega_i} would not hold is $\Omega_1 := \{x \in \real : \sin(1/x) \geq 0\}$ and $\Omega_2 := \{x \in \real : \sin(1/x) \leq 0 \}$.

The following theorem allows establishing $\delta$ISS of the system written in the form \eqref{eq:input-affine} by only checking $\delta$ISS of the individual modes.
This result can be seen as an extension of the discrete-time Demidovi\v{c} condition \cite[Theorem 14]{Tran2018Convergence} to the case of input-affine piecewise-smooth systems with continuous dynamics, where the proof is inspired by \cite[Theorem 2]{Pavlov2008Convergent}.

\begin{thm}
    Suppose that there exists a positive scalar $M_g < \infty$ such that $\|g_i(x)\| \leq M_g$, for all $x \in \real^n$ and all $i \in \{1, 2, \ldots, q\}$.
    Under Assumption~\ref{ass:Omega_i}, if there exist a positive definite matrix $Q \in \real^{n \times n}$ and a scalar $\alpha \in (0,1)$ such that the matrix
    \begin{align}
        \label{eq:Demidovic-like}
        \frac{\partial}{\partial x} F_i(x,w)^\top Q \frac{\partial}{\partial x} F_i(x,w) - \alpha^2 Q,
    \end{align}
    is negative semidefinite for all $(x, w) \in \Omega_i$ and for all $i \in \{1, \ldots, q\}$, then system \eqref{eq:input-affine} is $\delta$ISS.
\end{thm}
\begin{proof}
    \proofstep{From \eqref{eq:Demidovic-like} to $\delta$ISS of the individual mode:}
    This part of the proof adapts the arguments in the proof of \cite[Theorem 14]{Tran2018Convergence} to the case of input-affine, time-invariant dynamics.

    Let us denote, for brevity, $\delta x := x_1 - x_2$, $\delta w := w_1 - w_2$ and $\delta F_i := F_i(x_1,w_1) - F_i(x_2,w_2)$.
    Given two points $(x_1, w_1), (x_2, w_2) \in \Omega_i$, for some $i \in \{1, \ldots, q\}$, define the function $\Phi: [0,1] \to \real$ as follows:
    \begin{align}
        \Phi(s) := (\delta F_i)^\top Q F_i\left(s x_1 + (1-s) x_2, s w_1 + (1-s) w_2 \right),
    \end{align}
    and notice that
    \begin{align}
        \label{eq:demidovic-step1}
        \Phi(1) - \Phi(0) = (\delta F_i)^\top Q \delta F_i.
    \end{align}
    Due to the mean-value theorem, there exists $\bar s \in [0,1]$ such that
    \begin{align}
        \label{eq:demidovic-step2}
        \Phi(1) - \Phi(0) &= \frac{d}{ds} \Phi(\bar s) \\
        &= (\delta F_i)^\top Q \left[ \frac{\partial}{\partial x} F_i(\bar x, \bar w) \delta x + g_i(\bar x) \delta w \right], \nonumber
    \end{align}
    where we denoted $(\bar x, \bar w) := \bar s (x_1, w_1) + (1 - \bar s) (x_2, w_2)$.
    Equating the right-hand sides of \eqref{eq:demidovic-step1} and \eqref{eq:demidovic-step2}, and using the fact that, for any positive definite matrix, $$a^\top P a = a^\top P b \implies a^\top P a \leq b^\top P b,$$
    we obtain
    \begin{align}
        \label{eq:dISS-individual}
        \|\delta F_i\|_Q &\leq 
        \left\| \frac{\partial}{\partial x} F_i(\bar x, \bar w) \delta x + g_i(\bar x) \delta w \right\|_Q \nonumber \\
        &\leq \left\| \frac{\partial}{\partial x} F_i(\bar x, \bar w) \delta x \right\|_Q + \left\|g_i(\bar x) \delta w \right\|_Q \\
        &\leq \alpha \left\| \delta x \right\|_Q + \sqrt{\lambda_M(Q)} M_g |\delta w| =: \alpha \left\| \delta x \right\|_Q + \sigma |\delta w|, \nonumber
    \end{align}
    where the last inequality comes from \eqref{eq:Demidovic-like} and the uniform boundedness of the functions $g_i$.
    This proves that the individual mode dynamics satisfy a $\delta$ISS inequality.

    \proofstep{From the same mode to different modes:}
    Here we extend the proof of \cite[Theorem 2]{Pavlov2008Convergent}.
    Select $i,j \in \{1, \ldots, q\}$ such that $(x_1, w_1) \in \Omega_i$ and $(x_2, w_2) \in \Omega_j$.
    We only consider the case where $i \neq j$, as $i=j$ is considered in the first part of the proof.

    We consider the points $p^a := (x_1, w_1)$, $p^b := (x_2, w_2)$ and, as per Assumption~\ref{ass:Omega_i}, we identify a finite number $\ell \in \mathbb{N}_{>2}$ of points $p_k := (y_k, v_k), k \in \{1, 2, \ldots, \ell\}$, such that $p_1 = p^a$, $p_\ell = p^b$, and, for all $k \in \{1,2,\ldots,\ell-1\}$, $\overline{p_k p_{k+1}} \subset \Omega_{h(k)}$, for some $h(k) \in \{1,2,\ldots,q\}$.

    Due to the continuity of the dynamics, we have that $F_{h_k}(y_k, v_k) = F_{h_{k+1}}(y_k, v_k)$ for all $k \in \{1, \ldots, \ell-1\}$.
    Using this fact, the triangle inequality and \eqref{eq:dISS-individual}, we can write
        \begin{align}
            \label{eq:PWL-ISS-step1}
            \|F_i(x_1,w_1) - F_j(x_2,w_2)\|_Q
            &= \left\| \sum_{k = 1}^{\ell-1} F_{h_k}(y_{k+1}, v_{k+1}) - F_{h_k}(y_k, v_k) \right\|_Q \nonumber \\
            &\leq \sum_{k = 1}^{\ell-1} \left\| F_{h_k}(y_{k+1}, v_{k+1}) - F_{h_k}(y_k, v_k) \right\|_Q \nonumber \\
            &\leq \sum_{k = 1}^{\ell-1} \alpha \| y_{k+1} - y_k \|_Q + \sigma \left| v_{k+1} - v_k \right|.
        \end{align}
        We can now use the fact that $(y_k, v_k), \; k \in \{1, \ldots, \ell\},$ lie on the same line, and consequently
        \begin{align*}
            \sum_{k = 1}^{\ell-1} \alpha \| y_{k+1} - y_k \|_Q + \sigma \left| v_{k+1} - v_k \right| = \alpha \left\|x_1 - x_2 \right\|_Q + \sigma \left\| w_1 - w_2 \right\|,
        \end{align*}
        to conclude that
        \begin{align*}
            \|F_i(x_1,w_1) - F_j(x_2,w_2)\|_Q \leq \alpha\left\|x_1 - x_2 \right\|_Q + \sigma \left| w_1 - w_2 \right|,
        \end{align*}
        which shows that $V(x_1, x_2) := \|x_1 - x_2\|_Q$ is a $\delta$ISS Lyapunov function for the system \eqref{eq:input-affine} and therefore, due to \cite[Theorem 24]{Tran2016Incremental}, the system is $\delta$ISS.
\end{proof}

If the functions $f_i$ in \eqref{eq:input-affine} are linear and $g_i$ are constant matrices, then \eqref{eq:input-affine} corresponds to a piecewise-affine system
\begin{align}
    \label{eq:PWL}
    x = A_i x^- + B_i w^- + b_i, \quad (x^-, w^-) \in \Omega_i, \; i \in \{1, \ldots, q\}.
\end{align}
In this case, uniform boundedness of the functions $g_i$ is trivially satisfied, and \eqref{eq:Demidovic-like} simplifies to a system of linear matrix inequalities, as specified in the next corollary.
\begin{corollary}
    \label{cor:LMIs}
    If there exist a positive definite matrix $Q \in \real^{n \times n}$ and a scalar $\alpha$ such that
    \begin{align}
        \label{eq:LMIs}
        A_i^\top Q A_i \leq \alpha Q,
    \end{align}
    for all $i \in \{1, \ldots, q\}$, then system \eqref{eq:PWL} is $\delta$ISS.
\end{corollary}
Note that Corollary~\ref{cor:LMIs} has \cite[Theorem 2]{Pavlov2008Convergent} as a special case and extends it by allowing for switching also based on the external input.

\begin{example}{Linear plant + FOPE as a PWL system}
    We show now how we can obtain a system of the form \eqref{eq:PWL} from the interconnection of a linear plant
    \begin{align}
        \label{eq:linear_plant}
        \begin{split}
            x_p &= A x_p^- + B u^- + D w^-, \\
            e &= -C x_p,
        \end{split}
    \end{align}
    and the FOPE in \eqref{eq:FOPE}-\eqref{eq:Se-FOPE}.

    First, note that we can characterize the right-hand side of \eqref{eq:FOPE} as an input-and-state-dependent switched map.
    By using the shorthand notation $v := a u^- + b e$, such map is given by
    \begin{align}
        \label{eq:PWL_dyn}
        u = \Pi_{\K(e)}(v) = g(v,e) := \begin{cases}
            v, &\mbox{when } (v,e) \in \L, \\
            k_1 e, &\mbox{when } (v,e) \in \P_1, \\
            k_2 e, &\mbox{when } (v,e) \in \P_2,
        \end{cases}
    \end{align}
    where
    \begin{align}
        \label{eq:P_12}
        \begin{split}
            \L & := \left\{ (v,e) \in \real^2 : v \in \S(e) \right\}, \\
            \P_1 &:= \left\{ (v,e) \in \real^2 : v e \leq k_1 e^2 \right\}, \\
            \P_2 &:= \left\{ (v,e) \in \real^2 : v e \geq k_2 e^2 \right\}.
        \end{split}
    \end{align}
    Note that $\L$, $\P_1$, and $\P_2$ form a partition of $\real^2$.
    A graphical representation of this partition is reported in Figure~\ref{fig:partition} and the 3D plot of $g$ is depicted in Figure~\ref{fig:f_3D_plot} (notice that it is a continuous, piecewise-linear function).

    \begin{figure}[!tb]
        \centering
        \begin{minipage}[t]{0.47\textwidth}
            \centering
            \includegraphics[width=\linewidth]{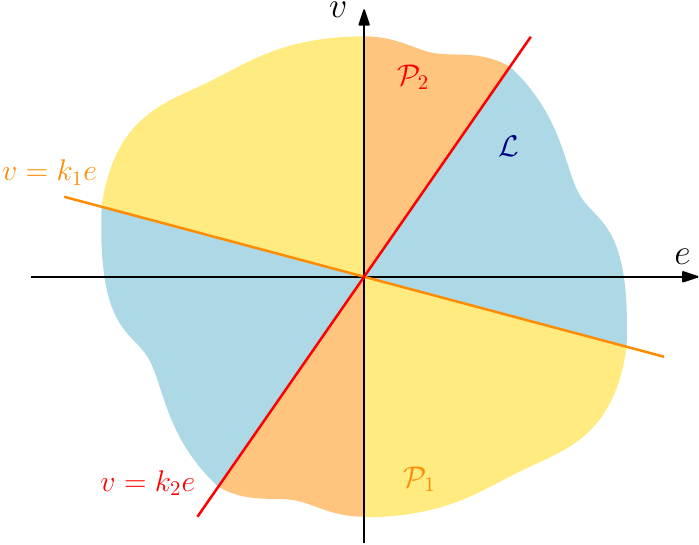}
            \caption{Graphical example of the partition of $\real^2$ into $\L, \P_1, \P_2$ in \eqref{eq:P_12}.}
            \label{fig:partition}
        \end{minipage}\hfill
        \begin{minipage}[t]{0.47\textwidth}
            \centering
            \includegraphics[width=\linewidth]{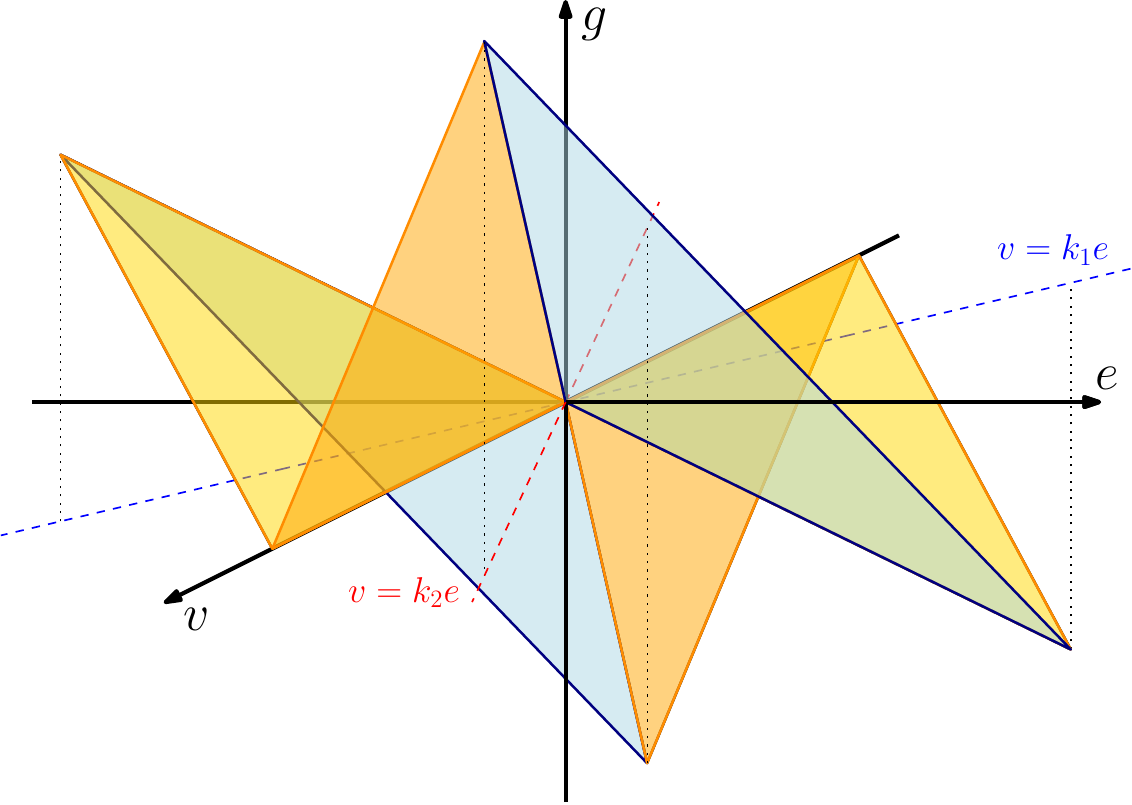}
            \caption{Graphical example of the plot of $g(v,e)$ in \eqref{eq:PWL_dyn}, using the partition in Figure~\ref{fig:partition}.}
            \label{fig:f_3D_plot}
        \end{minipage}
    \end{figure}

    Combining \eqref{eq:linear_plant}-\eqref{eq:P_12}, the closed-loop dynamics can be written in the form \eqref{eq:PWL}, with $q=3$, by defining the following matrices $A_i,B_i$ and sets $\Omega_i$:
    \begin{align}
        \label{eq:PWL-matrices}
        \begin{split}
            A_1 &:= \left[ \begin{array}{c|c}
                A & B \\ \hline -bCA & a - bCB
            \end{array} \right], \quad B_1 := \begin{bmatrix}
                F \\ -b C F
            \end{bmatrix} \\
            A_2 &:= \left[ \begin{array}{c|c}
                A & B \\ \hline -k_1 CA & -k_1 CB
            \end{array} \right], \quad B_2 := \begin{bmatrix}
                F \\ -k_1 C F
            \end{bmatrix} \\
            A_3 &:= \left[ \begin{array}{c|c}
                A & B \\ \hline -k_2 CA & -k_2 CB
            \end{array} \right], \quad B_3 := \begin{bmatrix}
                F \\ -k_2 C F
            \end{bmatrix} \\
            \Omega_1 &:= \left\{ (x^-, w^-) \in \real^{n+1+m} : k_1 e^2 \leq v e \leq k_2 e^2 \right\}, \\
            \Omega_2 &:= \left\{ (x^-, w^-) \in \real^{n+1+m} : v e \leq k_1 e^2 \right\}, \\
            \Omega_3 &:= \left\{ (x^-, w^-) \in \real^{n+1+m} : v e \geq k_2 e^2 \right\}.
        \end{split}
    \end{align}
    Note that, due to the linearity of the maps $(x, w) \mapsto e$ and $(x,w) \mapsto v$, the sets $\Omega_1, \Omega_2, \Omega_3$ are also closed, and they form a partition of $\real^{n+1+m}$.
\end{example}

\section{Numerical simulations}
\label{sec:simulation}

In this section, we show two numerical examples where the proposed approaches can certify convergence of the closed-loop including a FOPE.
In the first example, where the plant is nonlinear, we use the incremental gain of the FOPE, computed in the example in Section~\ref{subsec:inheritance}, to apply the small-gain condition in Theorem~\ref{thm:small-gain}.
In the second example, where the plant is linear, we also use the PWL representation \eqref{eq:PWL}, with the matrices and sets defined in \eqref{eq:PWL-matrices}, to check $\delta$ISS through the LMI condition \eqref{eq:LMIs}.

\subsection{Nonlinear plant}
\label{subsec:example-NL}

We begin by considering the following nonlinear plant with state $x = (x_1,x_2) \in \real^2$:
\begin{align}
    \label{eq:NL-numerical-state}
    x &= A x^- + E \begin{bmatrix}
        \sin \left( x_2^- \right) \\ \sin \left( x_1^- \right)
    \end{bmatrix} + B \mathrm{sat}_{[-1,1]} \left( u^- \right) + D w^- \\
    &:= \begin{bmatrix}
        0.5 & 0.2 \\ 0.1 & 0.4
    \end{bmatrix} x^- + 0.1 \begin{bmatrix}
        \sin \left( x_2^- \right) \\ \sin \left( x_1^- \right)
    \end{bmatrix} + \begin{bmatrix}
        0 \\ 0.1
    \end{bmatrix} \mathrm{sat}_{[-1,1]} \left( u^- \right) + \begin{bmatrix}
        0 \\ 0.1
    \end{bmatrix} w^- \nonumber, \\
    \label{eq:NL-numerical-output}
    e &= -Cx := - \begin{bmatrix}
        1 & 0
    \end{bmatrix} x,
\end{align}
where we also assumed that the control input is subject to physical constraints, modeled through the saturation function
\begin{align*}
    \mathrm{sat}_{[-\bar u, \bar u]} (x) := \max\{ \min \{ x, \bar u \}, -\bar u \}.
\end{align*}

Due to the nonlinearities in \eqref{eq:NL-numerical-state}, it is complicated to obtain an explicit piecewise-differentiable dynamics as in Section~\ref{sec:PW-dynamics}, therefore we will use the small-gain arguments of Theorem~\ref{thm:small-gain} to certify the convergence of the closed loop \eqref{eq:CL-NL} when the plant dynamics are given by \eqref{eq:NL-numerical-state} and the controller is a FOPE.

In order to verify whether the plant \eqref{eq:NL-numerical-state}-\eqref{eq:NL-numerical-output} is $\delta$ISS, we first note that both the sine function and the saturation function are Lipschitz continuous with Lipschitz constant $L=1$, thus
\begin{align*}
    \begin{cases}
        |\sin(x_1) - \sin(x_2)| \leq |x_1 - x_2|, \\
        |\mathrm{sat}_{[-1,1]}(x_1) - \mathrm{sat}_{[-1,1]}(x_2)| \leq |x_1 - x_2|.
    \end{cases}
\end{align*}
Using this fact, denoting $x_i(j) := x(j, \xi_i, (u_i, w_i)), \; i \in \{1,2\}$, we get, for any time $j \in \nat_{>0}$ and for all signals $u_1, u_2: \nat \rightarrow \real$, $w_1, w_2: \nat \rightarrow \real^m$,
\begin{align}
    \label{eq:NL-example-step}
    |\delta x(j)| := \; &|x_1(j) - x_2(j)| \nonumber \\
    \leq \; &|A| |\delta x(j-1)| + 0.1 |\delta x(j-1)| + |B| |\delta u(j-1)| \nonumber \\
    & \hspace{140pt} + |D| |\delta w(j-1)|, \nonumber \\
    \leq \; &\lambda |\delta x(j-1)| + 0.1 |\delta u(j-1)| + 0.1 |\delta w(j-1)|,
\end{align}
where we denoted all the incremental quantities with the $\delta$ prefix, for brevity, and we defined $\lambda := |A| + 0.1 = 0.7109$.
Iterating inequality \eqref{eq:NL-example-step}, we obtain that, for any pair of initial conditions $\xi_1, \xi_2 \in \real^2$,
\begin{align*}
    |\delta x(j)| \leq \lambda^j |\xi_1 - \xi_2| + \frac{0.1}{1-\lambda} \|\delta u\|_\infty + \frac{0.1}{1-\lambda} \|\delta w\|_\infty,
\end{align*}
which shows that \eqref{eq:NL-numerical-state} is $\delta$ISS with linear incremental gain $\gamma := \frac{0.1}{1-\lambda}$.
Since $e = - x_1$, we have $|\delta e| \leq |\delta x|$, so \eqref{eq:NL-numerical-state} is also $\delta$IOS with the same linear incremental gain.

Consequently, \eqref{eq:small-gain-sum} is satisfied by any selection of the FOPE parameters $a,b,k_1,j_2$ that satisfies the inequality
\begin{align*}
    \frac{b + \bar k}{1 - |a|} < \frac{1 - \lambda}{0.1} = 2.8912.
\end{align*}
To showcase the convergence property implied by $\delta$ISS, the closed loop composed by \eqref{eq:NL-numerical-state}-\eqref{eq:NL-numerical-output} and the FOPE was simulated starting from different initial conditions for the plant, $\xi_1 = (1,1)$, $\xi_2 = (-1,-1)$, and the same initial condition $u_1(0) = u_2(0) = 0$ for the FOPE.
The FOPE parameters were selected as $a = 0.8$, $b = 0.15$, $k_1 = 0$, $k_2 = 0.4$, which satisfy the small-gain condition.
The exogenous disturbance $w$ was selected identical for the two simulations, i.e., $w_1(j) = w_2(j) = \sin \left( \frac{2 \pi j}{20} \right)$.

Figure~\ref{fig:NL-plant-sim} compares the evolution of the plant states (top and middle plots) and the FOPE state/output (bottom plot) in the two simulations.
As expected from the uniform convergence property of the system, all the trajectories converge to each other, and the steady state trajectories have the same period as the exogenous input $w(j)$.

\begin{figure}[hbt]
    \centering
    \includegraphics[width=0.60\textwidth]{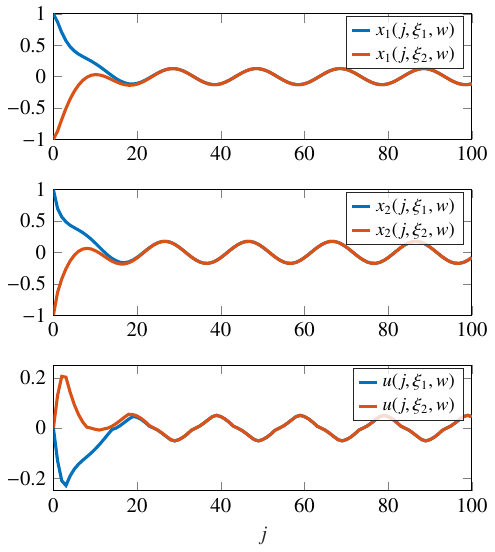}
    \caption{Evolution of the plant states (top and middle), as well as the controller output (bottom), evolving from different initial conditions, for system \eqref{eq:NL-numerical-state}-\eqref{eq:NL-numerical-output} connected in feedback with a FOPE.}
    \label{fig:NL-plant-sim}
\end{figure}

\subsection{Linear plant}
\label{subsec:example-linear}

For the second numerical example, we consider the same system simulated (in continuous time) in \cite[Section 5]{vdEijnden2023ASmall-Gain}.
The example considered there consisted of a second-order continuous-time LTI plant
\begin{align}
    \label{eq:plant_tf}
    \P(s) = \frac{1}{s^2 + \beta_0 \omega_0 s + \omega_0^2},
\end{align}
connected in feedback to an LTI controller $\mathcal C$
\begin{align}
    \label{eq:controller_tf}
    \mathcal C(s) = k_p \left( \frac{s + \omega_i}{s} \right) \left( \frac{\omega_{lp}^2}{s^2 + \beta_{lp} \omega_{lp} s + \omega_{lp}^2} \right).
\end{align}
A system like \eqref{eq:plant_tf} can arise in microelectromechanical nanopositioning applications \cite{Shi2022ANegative}, used in the lithography industry, and is thus of interest from a practical point of view.

In order to parallel the simulation results in \cite[Section 5]{vdEijnden2023ASmall-Gain}, where the incremental stability of a projection-based control loop is shown for the continuous-time case, we select the same numerical values for the plant parameters, namely $\omega_0 = 54 \cdot 2 \pi \; \mathrm{rad/s}$ and $\beta_0 = 0.009$, and the controller parameters, namely $k_p = 7 \cdot 10^4 \; \mathrm{N/m}$, $\omega_i = 4.75 \cdot 2 \pi \; \mathrm{rad/s}$, $\omega_{lp} = 6.5 \cdot 2 \pi \; \mathrm{rad/s}$, $\beta_{lp} = 0.8$.

Through exact discretization of the feedback interconnection of \eqref{eq:plant_tf}-\eqref{eq:controller_tf}, using a sampling time $T = 1 \; \mathrm{ms}$, we obtain the matrices $A,B,C,D$ in \eqref{eq:linear_plant}.

As the plant is linear, both the small-gain approach of Section~\ref{subsec:small-gain} and the LMI approach of Corollary~\ref{cor:LMIs} can be applied to this example.
We compare the two approaches below.

\paragraph{Small-gain approach}
An exponentially stable linear plant, like the one considered in this example, is also incrementally ISS and IOS with linear gains.
Specifically, the incremental gain can be identified as follows:
\begin{align*}
    \delta y(j) &= C A \delta x(j-1) + C B \delta u(j-1) + C D \delta w(j-1) \\
    &= C A^2 \delta x(j-2) + C A B \delta u(j-2) + C A D \delta w(j-1) \\ & \hspace{50pt} + C B \delta u(j-1) + C D \delta w(j-1) \\
    &= \ldots \\
    &= CA^j \delta \xi + \sum_{k=0}^j C A^k B \delta u(j-k) + \sum_{k=0}^j C A^k D \delta w(j-k).
\end{align*}
This yields the first inequality of \eqref{eq:dIOS-particular-sum}, with linear gain $\gamma_p^u = \|\P_\mathrm{dt}\|_{\ell_1} := \sum_{k = 0}^{\infty} C A^k B$, which is also known as the $\ell_1$-norm of the linear plant.
By means of numerical computations, we found the $\ell_1$-norm of the discretized plant to be $\|\P_\mathrm{dt}\|_{\ell_1} = 1.4191$.
Consequently, any set of the FOPE parameters that satisfies the inequality
\begin{align}
    \label{eq:small-gain-linear}
    \frac{b + \bar k}{1 - |a|} < \frac{1}{\|\P_\mathrm{dt}\|_{\ell_1}}
\end{align}
also satisfies the small-gain condition \eqref{eq:small-gain-sum}, thus it
ensures that the closed-loop system is $\delta$ISS and, consequently, convergent.

\paragraph{LMI approach}
After selecting the FOPE parameters, we can obtain the numerical values of the matrices $A_i, i \in \{1,2,3\},$ in \eqref{eq:PWL-matrices}, and use a numerical LMI solver, such as YALMIP \cite{Yalmip}, to look for a matrix $Q$ that solves the system of inequalities \eqref{eq:LMIs}.
If the solver identifies a feasible solution for the selected set of FOPE parameters, then the corresponding closed-loop system is $\delta$ISS and, consequently, convergent.
\begin{figure}[!tb]
    \centering
    \begin{minipage}[t]{0.47\textwidth}
        \centering
        \includegraphics[width=\linewidth]{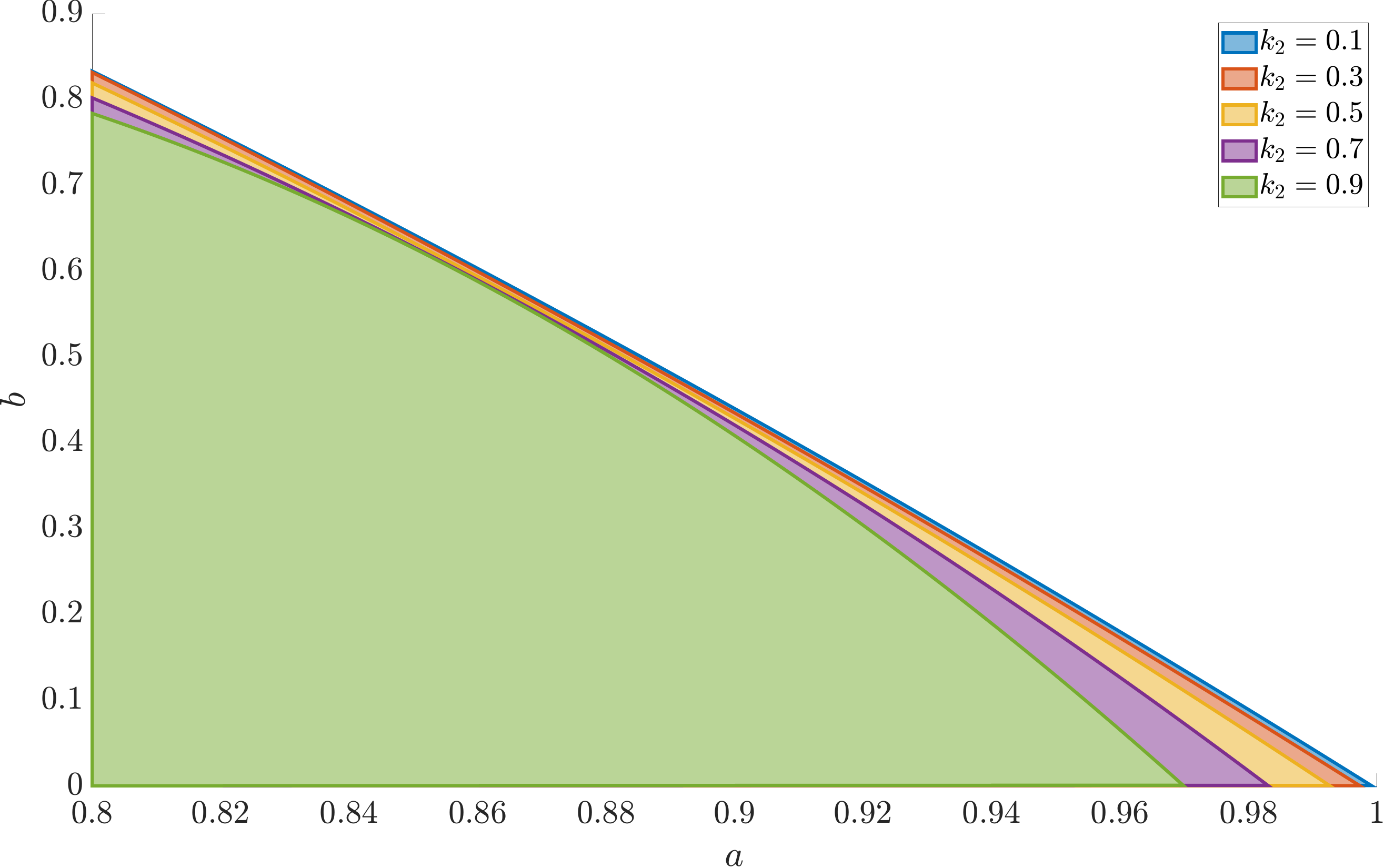}
        \caption{Feasibility of the system of LMIs \eqref{eq:LMIs} for the numerical example in Section~\ref{subsec:example-linear}, for $k_1 = 0$ and different values of $a$ (abscissas), $b$ (ordinates), $k_2$ (region color).}
        \label{fig:feasibility_common_P}
    \end{minipage}\hfill
    \begin{minipage}[t]{0.47\textwidth}
        \centering
        \includegraphics[width=\linewidth]{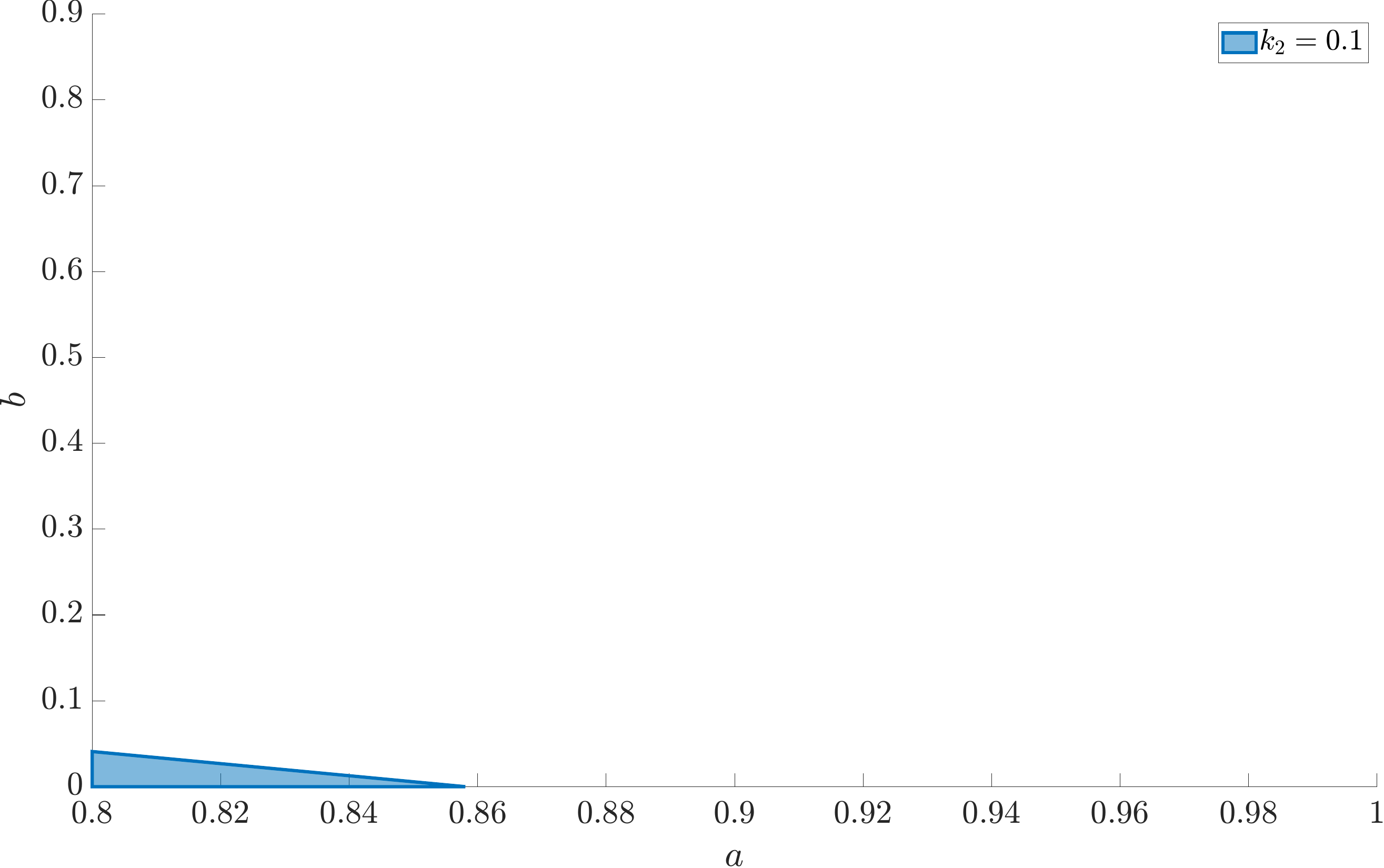}
        \caption{Set of parameters that verify the small-gain condition \eqref{eq:small-gain-linear} for the numerical example in Section~\ref{subsec:example-linear}, for $k_1 = 0$ and different values of $a$ (abscissas), $b$ (ordinates). Out of the 5 selected values of $k_2$, only the lowest produced a feasible region in the interval considered.}
        \label{fig:feasibility_small_gain}
    \end{minipage}
\end{figure}
\begin{figure}[htb]
    \centering
    \includegraphics[width=0.60\textwidth]{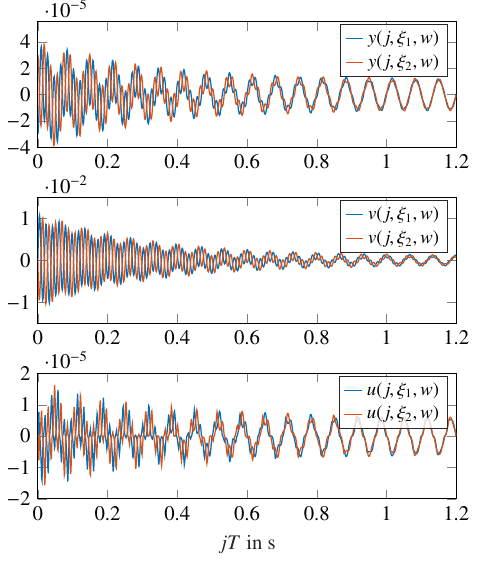}
    \caption{Evolution of the plant position (top) and velocity (middle), as well as the controller output (bottom), evolving from different initial conditions (blue and orange lines, respectively).}
    \label{fig:trajectories}
\end{figure}
\begin{figure}[htb]
    \centering
    \includegraphics[width=0.66\textwidth]{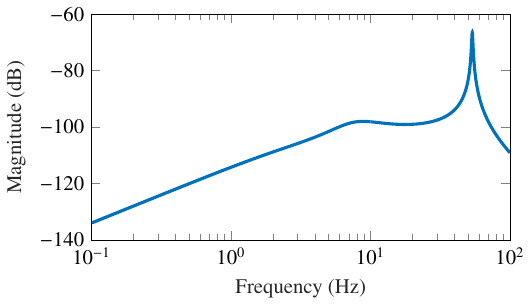}
    \caption{Nonlinear Bode-like plot of the transfer function gain for the closed-loop system \eqref{eq:PWL}.}
    \label{fig:BODE}
\end{figure}
\paragraph{Numerical results}
In order to compare the two approaches, we checked condition \eqref{eq:small-gain-linear} and the feasibility of the system of LMIs \eqref{eq:LMIs} for different values of the FOPE parameters, to check how conservative these numerical guarantees are with respect to each other.
In order to obtain representable results, we chose to keep one parameter constant, i.e. $k_1 = 0$, and swept through different values of the other parameters.
Specifically, we considered $a \in [0.8, 1]$, $b \in [0, 1]$ and $k_2 \in \{0.1, 0.3, 0.5, 0.7, 0.9\}$.
Figures~\ref{fig:feasibility_common_P}-\ref{fig:feasibility_small_gain} show a comparison between the regions of guaranteed $\delta$ISS, in the $a-b$ plane, for different values of $k_2$ (differently colored regions).
It appears evident that the LMI approach significantly reduces the conservatism when we know that the plant has a linear structure.
Given this fact, we selected a set of FOPE parameters that ensured the feasibility of the system of LMIs \eqref{eq:LMIs}, namely $a = 0.95$, $b = 0.1$, $k_1 = 0$, $k_2 = 0.5$, to simulate the evolution of the closed-loop system \eqref{eq:CL-NL} for this numerical example.

Figure~\ref{fig:trajectories} depicts the evolution of the plant position (top plot) and velocity (middle plot), as well as the controller output (bottom plot), evolving from different initial conditions $\xi_1, \xi_2$ (blue and orange lines, respectively), subject to the same disturbance $w(j) = \sin \left(15 \cdot \frac{2 \pi j}{1000} \right)$ (corresponding to the sampling of a continuous-time sinusoid with frequency $15 \mathrm{Hz}$).

As expected, after a transient characterized by fast oscillations, all the trajectories converge to each other, in view of the convergence property proved through the LMIs \eqref{eq:LMIs}.
Moreover, as expected, the steady-state evolution of the system trajectories is periodic, with the same frequency as the disturbance $w$.

In light of this last property, we assessed the input-output characteristics of the closed-loop system, by constructing a nonlinear Bode plot.
To this end, we excited the PWL closed-loop system \eqref{eq:PWL} with sinusoidal signals $w(j)$ of different frequencies, ranging from $10^{-1}$Hz to $10^{2}$Hz, obtaining outputs with the same period as the inputs.
Then, for each frequency, we computed the ratio between the root-mean-square value of the error $e$ and the signal $w$.
The results are depicted in Figure~\ref{fig:BODE}.
The possibility to obtain a Bode-like plot of the closed-loop system, even though it includes a nonlinear element such as the FOPE, allows for the objective evaluation of the controller performance (such as disturbance rejection or tracking), or its tuning using, e.g., frequency-domain techniques such as loop shaping.
This is not the focus of this paper, but it will be an interesting objective of future work.

\section{Conclusions}
In this paper, we have analyzed incremental stability and convergence properties of discrete-time projection-based control systems. We have shown that discrete-time projection-based controllers preserve quadratic incremental stability properties of their nominal dynamics, provided that the projection metric is well-designed. This key result has enabled the application of a small-gain theorem, for which we provided a formal proof for the first time, guaranteeing incremental stability of the feedback interconnection of a projection controller and a general nonlinear system. For input-affine piecewise smooth systems, we further developed a direct Lyapunov-based approach that generalizes discrete-time Demidovi\v{c} conditions to more general non-smooth settings. Continuity of the projection-based (closed-loop) dynamics is essential, which can be guaranteed for the projection controllers under suitable continuity conditions of the error-dependent constraint set. The results in this paper shed new light on incremental analysis of projection-based controllers, and open up new avenues for robust performance analysis using ``LTI-inspired'' frequency-domain tools. For future work, we will extend the results toward synthesis of projected control systems.

\smallskip

\textit{Acknowledgments:} We would like to thank Andrew R. Teel for the fruitful discussions and the help with the proof of Theorem~\ref{thm:small-gain}.
Moreover, we would like to thank the reviewers for the constructive feedback, which led us to significantly broaden the scope and improve the quality of the paper.

\FloatBarrier

\small
\bibliographystyle{unsrtnat}
\bibliography{refs}
\normalsize

\newpage

\appendix
    \section{Proofs of technical results}
    \label{app:proof-technical}

    \begin{proofof}{Proposition~\ref{prop:dISS-dIOS}}
        Let us denote $\delta x(j) := x(j, \xi_1, w_1) - x(j, \xi_2, w_2)$, $\delta y(j) := y(j, \xi_1, w_1) - y(j, \xi_2, w_2)$, $\delta w := w_1(j) - w_2(j)$.

        \smallskip

        \proofstep{$\delta$ISS $\implies$ $\delta$IOS + incremental detectability:}
        Incremental detectability is clearly implied by $\delta$ISS, by definition.
        Moreover, due to global Lipschitzness of $H$, there exists $L_Y$ such that, for all pairs of initial conditions $\xi_1, \xi_2 \in \real^{n_x}$ and all uniformly bounded input signals $w_1, w_2 : \nat \to \real^{n_w}$, it holds that\begin{align*}
            |\delta y(j)| &\leq L_Y |\delta x(j)|,
        \end{align*}
        which implies $\delta$IOS, due to the $\delta$ISS property.

        \smallskip

        \proofstep{$\delta$IOS + incremental detectability $\implies$ $\delta$ISS:}
        Let us first recall the \textit{weak triangular inequality} \cite[Equation (6)]{Jiang1994Small-gain}: for any $\alpha \in \K$ and any $\rho \in \K_\infty$, it holds that
        \begin{align}
            \label{eq:weak-triangle}
            \alpha(a + b) \leq \alpha \circ (\mathrm{Id} + \rho)(a) + \alpha \circ \left( \mathrm{Id} + \rho \right) \circ \rho^{-1}(b).
        \end{align}
        Let $\beta, \gamma$ be as in \eqref{eq:dIOS}, and let us rewrite the incremental detectability bound \eqref{eq:dDet} as follows:
        \begin{align}
            \label{eq:dDet-proof}
            |\delta x(j)| \leq \beta^0(|\xi_1 - \xi_2|, j) + \gamma^y \left( \|\delta y\|_\infty \right) + \gamma^w \left( \|\delta w\|_\infty \right).
        \end{align}
        From \eqref{eq:dIOS} we obtain the following upper bound:
        \begin{align*}
            \|\delta y\|_\infty \leq \beta(|\xi_1 - \xi_2|, 0) + \gamma \left( \|\delta w\|_\infty \right),
        \end{align*}
        and using this bound in \eqref{eq:dDet-proof} we get
        \begin{align}
            \label{eq:dISS-proof-bounded-x}
            \|\delta x\|_\infty &\leq \beta^0(|\xi_1 - \xi_2|, 0) + \gamma^y \left( \beta(|\xi_1 - \xi_2|, 0) + \gamma \left( \|\delta w\|_\infty \right) \right) \nonumber \\
            & \hspace{170pt} + \gamma^w \left( \|\delta w\|_\infty \right) \nonumber \\
            &\leq \sigma(|\xi_1 - \xi_2|) + \sigma^w \left( \| \delta w\|_\infty \right),
        \end{align}
        for some $\sigma, \sigma^w \in \K_\infty$, due to \eqref{eq:weak-triangle}.
        Moreover, due to time invariance and causality, \eqref{eq:dDet-proof} implies
        \begin{align}
            \label{eq:dISS-proof-step1}
            |\delta x(j)| \leq \beta^0(|\delta x(\floor{j/2})|, \floor{j/2}) + \gamma^y \left( \|\delta y\|_{[\floor{j/2}, j]} \right) + \gamma^w \left( \|\delta w\|_\infty \right).
        \end{align}
        For the incremental output, for every $k \in \{\floor{j/2}, \ldots, j\}$, \eqref{eq:dIOS} implies
        \begin{align}
            \label{eq:dISS-proof-step2}
            |\delta y(k)| &\leq \beta(|\xi_1 - \xi_2|, k) + \gamma \left( \|\delta w\|_\infty \right) \nonumber \\
            &\leq \beta(|\xi_1 - \xi_2|, \floor{j/2}) + \gamma \left( \|\delta w\|_\infty \right) \implies \nonumber \\
            &\hspace{-20pt} \implies \|\delta y(k)\|_{[\floor{j/2}, j]} \leq \beta(|\xi_1 - \xi_2|, \floor{j/2}) + \gamma \left( \|\delta w\|_\infty \right).
        \end{align}
        From \eqref{eq:dISS-proof-step1}, replacing $|\delta x(\floor{j/2})|$ with the upper bound in \eqref{eq:dISS-proof-bounded-x} and replacing $\|\delta y\|_{[\floor{j/2}, j]}$ with the upper bound in \eqref{eq:dISS-proof-step2}, we conclude $\delta$ISS (since a $\K\L$ function in $s,\floor{j/2}$ is also a $\K\L$ function in $(s,j)$).
    \end{proofof}

    \bigskip

    \begin{proofof}{Corollary~\ref{cor:dISS-convergence}}
        Since system \eqref{eq:gen_sys} is $\delta$ISS and $F(0,0) = 0$, by taking $\xi_2 = 0$, $w_2 = 0$ we conclude that \eqref{eq:gen_sys} is also (non-incrementally) ISS.
        Then, for every bounded input, in view of \cite[Theorem 1]{Jiang2001Input-to-state}, \eqref{eq:gen_sys} has an ISS Lyapunov function, i.e. there exist a function $V: \real^{n_x} \to \real$ and functions $\rho_1, \rho_2, \alpha, \sigma \in \K_\infty$ such that, for all $x \in \real^{n_x}$ and all $w \in \real^{n_w}$, it holds that
        \begin{align*}
            \rho_1(|x|) &\leq V(x) \leq \rho_2(|x|), \\
            V(F(x,w)) - V(x) &\leq - \alpha(V(x)) + \sigma(|w|).
        \end{align*}
        Let us denote by $M(w) < \infty$ the upper bound on $\|w\|_\infty$.
        Then, the set
        \begin{align*}
            \mathbb X := \{x \in \real^{n_x} : V(x) \leq \alpha^{-1} \circ \sigma(M(w)) \},
        \end{align*}
        is forward invariant for \eqref{eq:gen_sys}.
        %
        Due to the existence of a forward invariant set for any bounded $w$, continuity of $F$, and $\delta$UGAS of \eqref{eq:gen_sys} implied by $\delta$ISS, we can conclude, due to Proposition~\ref{prop:incremental-implies-contractive}, that system \eqref{eq:gen_sys} is uniformly globally convergent.
    \end{proofof}

    \bigskip

    \begin{proofof}{Lemma~\ref{lem:regularity-S}}
        Let us denote $f_{c,2} := f_c(x_c, e_2)$.
        \\
        \proofstep{Proof of (i) $\implies$ Ass.~\ref{ass:regularity-Se}:} Consider any point $\zeta \in \real^{n_c}$.
        Due to the definition of projection, we have that
        \begin{align}
            \label{eq:pure-translation-1}
            \left\| \zeta - \Pi^P_{\S(e)}(\zeta) \right\|_P \leq \left\| \zeta - z \right\|_P, \quad \mbox{for all } z \in \S(e).
        \end{align}
        Since $\S(e) = \S_0 + f_\S(e)$, there exists $y_0 \in \S_0$ such that $\Pi^P_{\S(e)}(\zeta) = y_0 + f_\S(e)$.
        Similarly, any $z \in \S(e)$ can be written as $z = z_0 + f_\S(e)$ for some $z_0 \in \S_0$.
        Substituting these equalities in \eqref{eq:pure-translation-1}, we get
        \begin{align*}
            \|\zeta - y_0 - f_\S(e)\|_P &\leq \|\zeta - z_0 - f_\S(e)\|_P, \quad \mbox{for all } z_0 \in \S_0.
        \end{align*}
        Rearranging the terms, we see that $y_0 = \Pi^P_{\S_0} (\zeta - f_\S(e))$, therefore
        \begin{align}
            \label{eq:pure-translation-2}
            \Pi^P_{\S(e)}(\zeta) = \Pi^P_{\S_0}(\zeta - f_\S(e)) + f_\S(e).
        \end{align}
        Moreover, the projection on a convex set is non-expansive with respect to $\|\cdot\|_P$ (see \cite[Proposition 4.8]{Bauschke2017}, replacing every scalar product $\langle \cdot, \cdot \rangle$ with $\langle \cdot, \cdot \rangle_P$).
        This implies, for any $e \in \real^{n_e}$ and any $z_1, z_2 \in \real^n_c$,
        \begin{align}
            \label{eq:non-expansive}
            \left| \Pi_{\S(e)}^P (z_1) - \Pi_{\S(e)}^P(z_2) \right|
            &\leq \sqrt{1/\lambda_m(P)} \left\|\Pi_{\S(e)}^P (z_1) - \Pi_{\S(e)}^P(z_2) \right\|_P \nonumber \\
            &\leq \sqrt{1/\lambda_m(P)} \left\| z_1 - z_2 \right\|_P \nonumber \\
            &\leq \sqrt{\lambda_M(P)/\lambda_m(P)} \left| z_1 - z_2 \right| \nonumber \\
            &= \sqrt{\kappa(P)} | z_1 - z_2 |.
        \end{align}
        Using \eqref{eq:pure-translation-2}-\eqref{eq:non-expansive}, the triangle inequality, and global Lipschitzness of $f_\S$, we conclude that, for all $e_1, e_2 \in \real^{n_e}$,
        \begin{align*}
            \left| \Pi^P_{\S(e_1)}(f_{c,2}) \right. &- \left. \Pi^P_{\S(e_2)}(f_{c,2}) \right| = \\
            &\hspace{-40pt}= \left| \Pi^P_{\S_0} (f_{c,2} - f_\S(e_1)) + f_\S(e_1) - \Pi^P_{\S_0} (f_{c,2} - f_\S(e_2)) - f_\S(e_2) \right| \\
            &\hspace{-40pt}\leq \left| \Pi^P_{\S_0} (f_{c,2} - f_\S(e_1)) - \Pi^P_{\S_0} (f_{c,2} - f_\S(e_2)) \right| + \left| f_\S(e_1) - f_\S(e_2) \right| \\
            &\hspace{-40pt}\leq \left( \sqrt{\kappa(P)} + 1 \right) \left| f_\S(e_1) - f_\S(e_2)\right| \leq L_\S \left( \sqrt{\kappa(P)} + 1 \right) \left| e_1 - e_2\right|,
        \end{align*}
        where $L_\S$ is the Lipschitz constant of $f_\S$.
        This proves the claim.
        \\
        \proofstep{Proof of (ii) $\implies$ Ass.~\ref{ass:regularity-Se}:} In this specific case, the projection $\Pi^P_{\S(e)}(f_{c,2})$ is the solution of the constrained QP
        \begin{align*}
            \Pi^P_{\S(e)}(f_{c,2}) = \arg\min_z &\|z\|_P \\
            \mbox{s.t. } & A z \leq A f_{c,2} + b(e).
        \end{align*}
        Recall that $f_{c,2}$ is a constant, so the constrained QP above is in the same form as \cite[eqn. (13)]{Bemporad02TheExplicit}, where the parameter $x(t)$ in \cite[eqn. (13)]{Bemporad02TheExplicit} corresponds to $b(e)$ in our case.
        Therefore, we know from \cite[Theorem 4]{Bemporad02TheExplicit} that the optimizer (i.e., the projection) is a continuous, piecewise-affine function of $b(e)$, namely $\Pi^P_{\S(e)}(x) = H_i b(e) + W_i$.
        Let us define $\bar H := \max_i \|H_i\|$.
        Then, we have that
        \begin{align*}
            \left| \Pi^P_{\S(e_1)}(f_{c,2}) - \Pi^P_{\S(e_2)}(f_{c,2}) \right| \leq \bar H |b(e_1) - b(e_2)| \leq \bar H L_b |e_1 - e_2|,
        \end{align*}
        where we denoted by $L_b$ the Lipschitz constant of $b$.
        \\
        \proofstep{Proof of (iii) $\implies$ Ass.~\ref{ass:regularity-Se}:}
        Let us denote $p_1 := \Pi^P_{\S(e_1)}(f_{c,2})$ and $p_2 := \Pi^P_{\S(e_2)}(f_{c,2})$.
        Due to the variational characterization of projections \cite[Theorem 3.14]{Bauschke2017}, we have, for all $s_1 \in \S(e_1)$,
        \begin{align*}
            \langle f_{c,2} - p_1, s_1 - p_1 \rangle_P \leq 0.
        \end{align*}
        Let $s_1$ be a point such that $|p_2 - s_1| \leq d_H(\S(e_1), \S(e_2))$.
        Then, using \eqref{eq:dH-Kinf}, we get
        \begin{align*}
                \langle f_{c,2} - p_1, p_2 - p_1 \rangle_P &= \langle f_{c,2} - p_1, p_2 - s_1 + s_1 - p_1 \rangle_P \\
                &\leq \langle f_{c,2} - p_1, p_2 - s_1 \rangle_P \\
                &\leq \lambda_M(P) \left| f_{c,2} - p_1 \right| \left| p_2 - s_1 \right| \\
                &\leq \lambda_M(P) |f_{c,2} - p_1| \alpha_H(|e_1 - e_2|).
        \end{align*}
        Lastly, using \eqref{eq:fc-unif-bounded}-\eqref{eq:Se-unif-bounded}, giving $|f_{c,2}| \leq M_c$ and $|p_1| \leq M_\S$, respectively, we conclude that
        \begin{align}
            \label{eq:unif-bounded-proof1}
            \langle f_{c,2} - p_1, p_2 - p_1 \rangle_P \leq \lambda_M(P) (M_c + M_\S) \alpha_H(|e_1 - e_2|)
        \end{align}
        With a symmetrical reasoning, we can conclude that
        \begin{align}
            \label{eq:unif-bounded-proof2}
            \langle f_{c,2} - p_2, p_1 - p_2 \rangle_P \leq \lambda_M(P) (M_c + M_\S) \alpha_H(|e_1 - e_2|).
        \end{align}
        Summing \eqref{eq:unif-bounded-proof1} and \eqref{eq:unif-bounded-proof2}, we get
        \begin{align*}
            | p_1 - p_2 |^2 &\leq 1/\lambda_m(P) \langle p_1 - p_2, p_1 - p_2 \rangle_P \\
            &\leq 2 \kappa(P) (M_c + M_\S) \alpha_H(|e_1 - e_2|).
        \end{align*}
        Taking the square root of both sides, we obtain \eqref{eq:LH}.
    \end{proofof}

    \bigskip

    \begin{proofof}{Lemma~\ref{lem:Lipschitz}}
        Let us denote, for brevity, $f_{c,i} := f_c(x_i, e_i), \; i \in \{1,2\}$.
        Using the triangle inequality, \eqref{eq:LH}, and \eqref{eq:non-expansive}, we have, for all $(x_1,e_1), (x_2,e_2) \in \real^{n_c} \times \real^{n_e}$,
        \begin{align*}
            \left| \Pi^P_{\S(e_1)} f_{c,1} \right. & \left.- \hspace{4pt} \Pi^P_{\S(e_2)} f_{c,2} \right| = \\
            &= \left| \Pi^P_{\S(e_1)} f_{c,1} - \Pi^P_{\S(e_1)} f_{c,2} + \Pi^P_{\S(e_1)} f_{c,2} - \Pi^P_{\S(e_2)} f_{c,2} \right| \\
            &\leq \left| \Pi^P_{\S(e_1)} f_{c,1} - \Pi^P_{\S(e_1)} f_{c,2} \right| + \left| \Pi^P_{\S(e_1)} f_{c,2} - \Pi^P_{\S(e_2)} f_{c,2} \right| \\
            &\leq \sqrt{\kappa(P)} |f_{c,1} - f_{c,2}| + \alpha_{\mathcal S}(|e_1 - e_2|).
        \end{align*}
        Due to the continuity of $f_c$, as per Assumption~\ref{ass:regularity}, the previous inequality proves continuity of $\Pi^P_{\S(e)}(f_c)$ and, combined with the continuity of $f_p$, of $f$ in \eqref{eq:CL-NL}.
    \end{proofof}

    \section{Proof of Theorem~\ref{thm:small-gain}}
    \label{app:proof-small-gain}
    In order to prove Theorem~\ref{thm:small-gain}, we will need the following lemma, which is an adaptation of \cite[Lemma A.1]{Jiang1994Small-gain}, where we use the $\max$ on the right-hand side, rather than the summation.
    Given an essentially bounded function $z: \real_{\geq 0} \to \real_{\geq 0}$, we will use the notation $\|z\|_{[t_1, t_2]} := \sup_{t \in [t_1, t_2]} |z(t)|$.
    \begin{lemma}
    \label{lem:technical}
        Let $\mu \in (0,1]$, $\beta \in \K\L$, and $\gamma \in \K_\infty$ be such that $\gamma(s) < s$ for all $s > 0$.
        Then, there exists $\hat \beta \in \K\L$ such that, for any $s, d \in \real_{\geq 0}$ and any non-negative function $z: \real_{\geq 0} \to \real_{\geq 0}$, defined and essentially bounded on $[0, \infty)$ and satisfying
        \begin{align}
            \label{eq:lemma-hp}
            z(t) \leq \max \left\{ \beta(s,t), \gamma(\|z\|_{[\mu t, \infty)}), d \right\}, \quad \mbox{for all } t \in \real_{\geq 0},
        \end{align}
        we have
        \begin{align}
            \label{eq:lemma-claim}
            z(t) \leq \max \left\{ \hat \beta(s,t), d \right\}, \quad \mbox{for all } t \in \real_{\geq 0}.
        \end{align}
    \end{lemma}
    \begin{proof}
        We want to prove that, for each $r, \varepsilon > 0$, there exists a time $T(r,\varepsilon) > 0$ such that, if \eqref{eq:lemma-hp} holds with $s \leq r$, then
        \begin{align}
            \label{eq:lemma-claim-d0}
            z(t) \leq \max\{\varepsilon, d\}, \quad \mbox{for all } t \geq T(r,\varepsilon).
        \end{align}
        Indeed, if the signal $z(t)$ eventually converges to a region bounded by an arbitrarily small $\varepsilon$ (or the constant $d$), it can be globally upper-bounded by a new $\mathcal{KL}$ function, $\hat{\beta}(s,t)$, combined with $d$, i.e. \eqref{eq:lemma-claim}.
        To proceed, we define a sequence of times $\{t_n\}_{n \in \mathbb N}$ such that $t_0=0$ and, for $n \geq 1$, $t_n$ is the smallest real number such that
        \begin{align}
        \label{eq:sequence}
            \beta(r,t_n) \leq \gamma^{n}(\beta(r,0)).
        \end{align}
        This sequence tracks the time required for the $\mathcal{KL}$ function $\beta$ to decay below successive compositions of the $\mathcal{K}$ function $\gamma$.
        Note that such a sequence always exists, since $\beta \in \mathcal{KL}$ and $\gamma \in \mathcal{K}$.
        Then, we define an auxiliary sequence $\{ \bar t_n \}_{n \in \mathbb N}$ as follows:
        \begin{align}
            \label{eq:new-sequence}
            \begin{cases}
                \bar t_0 = 0, \\
                \bar t_{n+1} = \max \{t_{n+1}, \bar{t}_{n}/\mu \}.
            \end{cases}
        \end{align}
        
        We will now prove by induction that, if \eqref{eq:lemma-hp} holds with $s \leq r$, then, for each $n$,
        \begin{align}
            \label{eq:equivalent-bound}
            z(t) \leq \max\left\{\gamma^{n}(\beta(r,0)), d \right\}, \quad \mbox{for } t \geq \bar{t}_n.
        \end{align}
        This will show that the upper bound of $z(t)$ strictly decreases over successive time domains defined by $\overline{t}_n$. As $n$ increases, $\gamma^n$ shrinks the bound, eventually fulfilling the convergence requirement \eqref{eq:lemma-claim-d0} for a sufficiently large $n$. To proceed by induction,
        note that \eqref{eq:equivalent-bound} holds with $n=0$, since $\gamma^{0}:=\mathrm{Id}$.
        Now suppose that \eqref{eq:equivalent-bound} holds for some $n \geq 0$.
        Then, due to \eqref{eq:new-sequence}, $\bar t_{n+1} \geq t_{n+1}$, which, combined with \eqref{eq:sequence} and $s \leq r$, implies that, for all $t \geq \bar t_{n+1}$,
        \begin{align}
            \label{eq:lemma2-proof-step1}
            \beta(s,t) \leq \beta(r,t) \leq \beta(r,\bar t_{n+1}) \leq \gamma^{n+1} (\beta(r,0)).
        \end{align}
        Moreover, again due to \eqref{eq:new-sequence}, $\mu \bar t_{n+1} \geq \bar t_n$, which implies that, for all $t \geq \bar t_{n+1}$,
        \begin{align}
            \label{eq:lemma2-proof-step2}
            \begin{split}
                \gamma(\|z\|_{[\mu t, \infty)}) \leq \gamma(\|z\|_{[\mu \bar t_{n+1}, \infty)}) &\leq \gamma(\max\left\{\gamma^n(\beta(r,0), d\right\}) \\
                &= \max\left\{ \gamma^{n+1} (\beta(r,0)), \gamma(d) \right\} \\
                &< \max \left\{ \gamma^{n+1}(\beta(r,0)), d \right\},
            \end{split}
        \end{align}
        where the last inequality is due to $\gamma(s) < s$.
        Using \eqref{eq:lemma2-proof-step1}, \eqref{eq:lemma2-proof-step2} in \eqref{eq:lemma-hp} we get that, for $t \geq \bar t_{n+1}$,
        \begin{align*}
            z(t) &\leq \max \left\{ \gamma^{n+1}(\beta(r,0)), \gamma^{n+1}(\beta(r,0)), d, d \right\} \\
            &= \max \left\{ \gamma^{n+1}(\beta(r,0)), d \right\},
        \end{align*}
        which proves \eqref{eq:equivalent-bound}, and thus the existence of $\hat \beta \in \K\L$ such that \eqref{eq:lemma-hp} implies \eqref{eq:lemma-claim}, completing the proof.
    \end{proof}
    Using this lemma, we can proceed to prove the theorem.

    \medskip

    \begin{proofof}{Theorem~\ref{thm:small-gain}}
        We first prove uniform boundedness of $\delta e, \delta u$.
        Due to causality, $\delta$IOS is equivalent to
        \begin{subequations}
            \begin{align}
                \label{eq:causality-y1}
                |\delta e(j)| &\leq \max \left\{ \beta_1(|\delta \xi_p|, j), \gamma_1(\|\delta u\|_{[0,j]}), \gamma^w(\|\delta w\|_\infty) \right\}, \\
                \label{eq:causality-y2}
                |\delta u(j)| &\leq \max \left\{ \beta_2(|\delta \xi_c|, j), \gamma_2(\|\delta e\|_{[0,j]})\right\},
            \end{align}
        \end{subequations}
        where we used the interconnection law in Assumption 3. Note that $\|\delta e\|_{[0,j]}, \|\delta u\|_{[0,j]}$ are bounded, because there is no finite escape time in discrete-time dynamical systems.
        Moreover, since $\beta_1, \beta_2 \in \mathcal{KL}$,
        \begin{subequations}
            \begin{align}
                \label{eq:bounded-y1-step1}
                \|\delta e\|_{[0,j]} &\leq \max \left\{ \beta_1(|\delta \xi_p|, 0), \gamma_1(\|\delta u\|_{[0,j]}), \gamma^w(\|\delta w\|_\infty) \right\}, \\
                \label{eq:bounded-y2-step1}
                \|\delta u\|_{[0,j]} &\leq \max \left\{ \beta_2(|\delta \xi_c|, 0), \gamma_2(\|\delta e\|_{[0,j]}) \right\}.
        \end{align}
        \end{subequations}
        Using \eqref{eq:bounded-y2-step1} in \eqref{eq:bounded-y1-step1}, we get
        \begin{align*}
            \|\delta e\|_{[0,j]} \leq \max &\left\{ \beta_1(|\delta \xi_p|, 0), \gamma_1 \circ \beta_2(|\delta \xi_c|, 0), \right.\\
            &\hspace{50pt} \left. \gamma_1 \circ \gamma_2(\|\delta e\|_{[0,j]}), \gamma^w(\|\delta w\|_\infty) \right\}.
        \end{align*}
        With a symmetrical reasoning for $\delta u$, using the small-gain condition \eqref{eq:small-gain} and taking the limit for $j \to \infty$, we obtain the following uniform bound on the incremental outputs:
        \begin{align}
            \label{eq:unif-bound-y1}
            \|\delta e\|_\infty &\leq \max \left\{ \beta_1(|\delta \xi_p|, 0), \gamma_1 \circ \beta_2(|\delta \xi_c|, 0), \gamma^w(\|\delta w\|_\infty) \right\}, \nonumber \\
            &=: \max \left\{ \hat \alpha_1 \left( \left| \delta \xi \right| \right), \hat \gamma_1^w \left( \| \delta w\|_\infty \right) \right\}, \\
            \label{eq:unif-bound-y2}
            \|\delta u\|_\infty &\leq \max \left\{ \beta_2(|\delta \xi_c|, 0), \gamma_2 \circ \beta_1(|\delta \xi_p|, 0), \gamma_2 \circ \gamma^w(\|\delta w\|_\infty) \right\}, \nonumber \\
            &=: \max \left\{ \hat \alpha_2 \left( \left| \delta \xi \right| \right), \hat \gamma^w \left( \| \delta w\|_\infty \right) \right\},
        \end{align}
        for some $\hat \alpha_i, \hat \gamma_i^w \in \mathcal{K}_\infty$.

        Next, we compute a uniform $\delta$IOS bound for $\delta e$.
        Consider again \eqref{eq:dIOS-particular}, and note that, due to causality, time invariance, $\beta_1 \in \mathcal{KL}$, and the fact that $j - \floor{j/2} \geq \floor{j/2}$ for all $j \in \mathbb N$,
        \begin{align}
            \label{eq:KL-step1}
            |\delta e(j)| 
            \leq \max &\left\{ \beta_1 \left( \left| \delta x_p \left( \floor{j/2} \right) \right|, \floor{j/2} \right), \right. \\
            &\hspace{50pt} \left. \gamma_1 \left( \|\delta e\|_{[\floor{j/2},j]} \right), \gamma^w \left( \|\delta w\|_\infty \right) \right\}. \nonumber
        \end{align}
        Now, for every $k \in \{\floor{j/2}, \ldots, j \}$, time invariance, \eqref{eq:causality-y1} and $\beta_1 \in \mathcal{KL}$ give
        \begin{align*}
            |\delta u(k)| &\leq \max \left\{ \beta_2(|\delta x_c(\floor{j/4})|, k - \floor{j/4}), \gamma_2(\|\delta e\|_{[\floor{j/4},k]}) \right\} \\
            &\leq \max \left\{ \beta_1 \left(|\delta \xi_c|, \floor{j/4} \right), \gamma_2 \left(\|\delta e\|_{[\floor{j/4},\infty]} \right) \right\},
        \end{align*}
        which implies
        \begin{align*}
            \|\delta u\|_{[\floor{j/2}, j]} \leq \max \left\{ \beta_1 \left(|\delta \xi_c|, \floor{j/4} \right), \gamma_2 \left(\|\delta e\|_{[\floor{j/4},\infty]} \right) \right\}.
        \end{align*}
        Using this bound in \eqref{eq:KL-step1}, we get
        \begin{align}
            \label{eq:KL-step2}
            |\delta e(j)| \leq \max &\left\{ \beta_1 \left( \left| \delta x_p \left( \floor{j/2} \right) \right|, \floor{j/2} \right), \gamma_1 \circ \beta_2(|\delta \xi_c|, \floor{j/4}), \right. \nonumber \\
            &\hspace{20pt} \left. \gamma_1 \circ \gamma_2 \left( \|\delta e\|_{[\floor{j/4}, \infty]} \right), \gamma^w \left( \|\delta w\|_\infty \right) \right\},
        \end{align}
        Due to incremental detectability \eqref{eq:dDet} and uniform boundedness of $\delta e, \delta u$, proved in \eqref{eq:unif-bound-y1}-\eqref{eq:unif-bound-y2}, there exist functions $\bar \beta_2 \in \mathcal{KL}$ and $\bar \sigma_2 \in \mathcal{K}_\infty$ such that
        \begin{align}
            \label{eq:KL-step3}
            \beta_1 \left( \left| \delta x_p \left( \floor{j/2} \right) \right|, \floor{j/2} \right)
            &\leq \max \left\{ \bar \beta_1 \left( | \delta \xi |, j \right), \bar \sigma_1 \left( \|\delta w\|_\infty \right) \right\}.
        \end{align}
        Substituting \eqref{eq:KL-step3} in \eqref{eq:KL-step2}, we get that there exist $\hat \beta_1 \in \mathcal{KL}$ and $\hat \sigma_1 \in \mathcal K$ such that
        \begin{align}
            |\delta e(j)| \leq \max \left\{ \hat \beta_1 \left( | \delta \xi |, j \right), \gamma_1 \circ \gamma_2 \left( \|\delta e\|_{[\floor{j/4}, \infty]} \right), \hat \sigma_1 \left( \|\delta w\|_\infty \right) \right\}.
        \end{align}
        Finally, we can use the Lemma~\ref{lem:technical} and a ``zero-order-held version'' of $\delta e$ to prove a $\delta$IOS bound on $\delta e$.
        By a symmetrical procedure, an equivalent bound can be obtained for $\delta u$, proving that the closed loop is $\delta$IOS.

        Since incremental detectability of both the plant and controller implies incremental detectability of their closed-loop interconnection, we can use Proposition~\ref{prop:dISS-dIOS} to conclude that \eqref{eq:CL-NL} is $\delta$ISS.
    \end{proofof}

\end{document}